%% file: main_arxiv.tex
\documentclass{article}
\usepackage{amsmath,amssymb,enumerate,mathrsfs}
\usepackage[ruled,linesnumbered,algo2e,vlined]{algorithm2e}
\SetArgSty{textup}
\usepackage{amsthm}
\usepackage[caption=false]{subfig}
\usepackage{environ}
\NewEnviron{mysubfig}[2]{%
	\begin{minipage}{#1}\centering\subfloat[#2]{\centering\BODY}\end{minipage}%
}
\usepackage{tikz,graphicx}
\usetikzlibrary{arrows,snakes,backgrounds}

 \usepackage{latexsym}

\def\Underline{\setbox0\hbox\bgroup\let\\\endUnderline}
\def\endUnderline{\vphantom{y}\egroup\smash{\underline{\box0}}\\}
\def\|{\verb|}

\theoremstyle{definition}
\newtheorem{definition}{Definition}{\bfseries}{\rmfamily}
\newtheorem{example}{Example}{\bfseries}{\rmfamily}
\newtheorem{condition}{Condition}{\bfseries}{\rmfamily}
\theoremstyle{theorem}
\newtheorem{theorem}{Theorem}{\bfseries}{\rmfamily}
\newtheorem{lemma}{Lemma}{\bfseries}{\rmfamily}
\newtheorem{corollary}{Corollary}{\bfseries}{\rmfamily}

\newcommand{\C}{\mathcal{C}}
\newcommand{\tw}{\mathit{tw}}

\newcommand{\prl}[3]{#2\mathrel{/\!/\!}_{#1}#3}
\newcommand{\its}[3]{#2\mathrel{\text{$/$\llap{$\backslash$}}}_{#1}#3}

\newcommand{\mscr}[1]{\mathscr{#1}}
\newcommand{\mcal}[1]{\mathcal{#1}}
\newcommand{\mrm}[1]{\mathrm{#1}}

\newcommand{\msf}[1]{\mathsf{#1}}
\newcommand{\mbb}[1]{\mathbb{#1}}
\newcommand{\mtt}[1]{\mathtt{#1}}

\newcommand{\suc}[1]{\msf{suc}_{#1}}
\newcommand{\pred}[1]{\msf{pred}_{#1}}
\newcommand{\point}[3][]{% [option],x,y の3値
	\draw[#1] (#2,#3-0.2) -- (#2,#3+0.2);
}
\newcommand{\interval}[4][]{% [option],l,r,y の4値
	\draw[#1] (#2,#4) -- (#3,#4);
}

\newcommand{\lrinterval}[4][]{% [option],l,r,y の4値
	\interval[#1]{#2}{#3}{#4}
	\point[#1]{#2}{#4}
	\point[#1]{#3}{#4}
}
\newcommand{\rtdlritvl}[7][]{%
	\lrinterval[#1]{#2}{#3}{#4}
	\draw[#1, densely dotted] (#2,#4) -- (#2,#5);
	\draw[#1, densely dotted] (#3,#4) -- (#3,#5);
	\point[#1,thick]{#2}{#5}
	\point[#1,thick]{#3}{#5}
	\draw (#2,#5-0.5) node{#6};
	\draw (#3,#5-0.5) node{#7};
}
\newcommand{\rtdlritvlf}[7][]{%
	\fill[red!8] (#2,#4) -- (#2,#5) -- (#3,#5) -- (#3,#4);
	\lrinterval[-, thick, densely dotted, color=red]{#2}{#3}{#4}
	\draw[-, densely dotted, color=red] (#2,#4) -- (#2,#5);
	\draw[-, densely dotted, color=red] (#3,#4) -- (#3,#5);
	\draw (#2,#5-0.5) node[color=red]{#6};
	\draw (#3,#5-0.5) node[color=red]{#7};
}

\tikzstyle{forbidden}=[style=thick, color=red]
\tikzstyle{forgotten}=[style=thick, densely dashed, color=red]
\tikzstyle{focused}=[style=thick, color=blue]
\tikzstyle{finterval}=[-, thick, snake=snake, segment amplitude=1.5pt, segment length=4pt, color=red]

\begin{document}

%\title{区間グラフ辺削除に対する木幅パラメータの固定パラメータアルゴリズム}

\title{Fixed-Treewidth-Efficient Algorithms for Edge-Deletion to Interval Graph Classes\thanks{This manuscript corrects some errors in~\cite{SaitohYB21}.}}

\author{Toshiki Saitoh\thanks{Kyushu Institute of Technology, \texttt{toshikis@ces.kyutech.ac.jp}}
 \and Ryo Yoshinaka\thanks{Tohoku University, \texttt{ryoshinaka@tohoku.ac.jp}}
 \and Hans L. Bodlaender\thanks{Utrecht University, \texttt{H.L.Bodlaender@uu.nl}}}

\date{}

\maketitle

\begin{abstract}
  For a graph class $\C$, the $\C$-\textsc{Edge-Deletion} problem asks for a given graph $G$ to delete the minimum number of edges from $G$ in order to obtain a graph in $\C$. We study the $\C$-\textsc{Edge-Deletion} problem for $\C$ the class of interval graphs and other related graph classes.
  It follows from Courcelle's Theorem that these problems are fixed parameter tractable when parameterized by treewidth. In this paper, we present concrete FPT algorithms for these problems. By giving explicit algorithms and analyzing these in detail, we obtain algorithms that are significantly faster than the algorithms obtained by using Courcelle's theorem.
\end{abstract}

\input{intro.tex}
\input{preliminaries.tex}
\input{interval.tex}
\input{permutation.tex}

\input{interval_variants.tex}
%\input{circle.tex}
\input{threshold.tex}

\input{conclusion.tex}

\subsubsection*{Acknowledgement}
%\noindent {\bf Acknowledgement: }
We are very much grateful to Yasuaki Kobayashi for his valuable comment that has improved the complexity analyses.
This work was supported in part by JSPS KAKENHI Grant Numbers JP18H04091 and JP19K12098.

\bibliographystyle{plain}
\bibliography{ref}

% \newpage
% \input{appendix}

\end{document}

%% file: intro.tex
% !TEX root = main_arxiv.tex
\section{Introduction}

%%%%
%% About geometric graphclasses (注: 今回は perfect graphとは限らない)
Intersection graphs are represented by geometric objects aligned in certain ways so that each object corresponds to a vertex and two objects intersect if and only if the corresponding vertices are adjacent. 
Intersection graphs are well-studied in the area of graph algorithms since there are many important applications and we can solve many NP-hard problems in general graphs in polynomial time on such graph classes. 
Interval graphs are intersection graphs which are represented by intervals on a line.
\textsc{Clique}, \textsc{Independent Set}, and \textsc{Coloring} on interval graphs can be solved in linear time and interval graphs have many applications in bioinformatics, scheduling, and so on. 
%such problems on interval graphs also can be solved in polynomial time since interval graphs are subclass of chordal graphs. % グラフの包含関係の話はいる？
See~\cite{Golumbic:2004,Brandstadt:1999,Spinrad} for more details of interval graphs and other intersection graphs.

% Graph modification problem 
Graph modification problems on a graph class $\C$ are to find a graph in $\C$ by modifying a given graph in certain ways. 
%There are some kinds of the modifications such that deleting vertices, deleting edges, and adding edges. 
$\C$-\textsc{Vertex-Deletion}, $\C$-\textsc{Edge-Deletion}, and $\C$-\textsc{Completion} are to find a graph in $\C$ by deleting vertices, deleting edges, and adding edges, respectively,  with the minimum cost. 
These problems can be seen as generalizations of many NP-hard problems. 
\textsc{Clique} is equivalent to \textsc{Complete-Vertex-Deletion}: we find a complete graph by deleting the smallest number of vertices.
%% \textsc{Feedback Vertex Set}: \textsc{Forest-Vertex Deletion}.
%% \textsc{Hamiltonian Path} is \textsc{Path-Edge Deletion}.
%
%Modification problems on such graph classes cannot work on the relationship for graph classes.
%問題の難しさがグラフクラスの包含関係とは独立であることを言いたい．（必要？？？）
% (interval) edge deletion problem
Modification problems on intersection graph classes also have many applications.
For example, \textsc{Interval-Vertex/Edge-Deletion} problems have applications to DNA (physical) mapping~\cite{GoldbergGKS95,FrenkelPMFK10,WatermanG86}. 
%\textsc{Bandwidth} and \textsc{Minimum Fill-in} is \textsc{Proper Interval-Completion} and \textsc{Chordal-Completion}, respectively.
%interval deletion の応用について詳しく触れる？？？
%
%
%However, it is well-known that the problems $\C$-\textsc{Vertex Deletion}, $\C$-\textsc{Edge Deletion}, and $\C$-\textsc{Completion} for many graph classes $\C$ including interval graphs are NP-hard~\cite{Yannakakis81,LewisY80,Garey:1979,GoldbergGKS95}. % 27, 13, 23, 20 (removed , Yannakakis81a,Yannakakis78(conference))
%Vertex Deletion
Lewis and Yannakakis showed that $\C$-\textsc{Vertex-Deletion} is NP-complete for any nontrivial hereditary graph class~\cite{LewisY80}. A graph class $\C$ is hereditary if for any graph in $\C$, every induced subgraph of the graph is also in $\C$.
Since the class of intersection graphs are hereditary, $\C$-\textsc{Vertex Deletion} is NP-complete for any nontrivial intersection graph class $\C$. %
The problems $\C$-\textsc{Edge-Deletion} are also NP-hard when $\C$ is the class of perfect, chordal, split, circular arc, chain~\cite{NatanzonSS01}, interval, proper interval~\cite{GoldbergGKS95}, trivially perfect (a.k.a.\ nested interval)~\cite{SharanThesis}, threshold~\cite{Margot94}, permutation, weakly chordal, or circle graphs~\cite{BurzynBD06}.
%Yannakakis~\cite{Yannakakis81} also have shown that $\C$-\textsc{Edge Deletion} are NP-hard for many classes $\C$ which consist of graphs without cycles of specified length $l$, or of any length $ \leqq l$, connected and degree-constrained graphs, outerplanar graphs, transitive digraphs, line-invertible graphs, bipartite graphs, or transitively orientable graphs. 
See the lists in~\cite{ManciniThesis,BurzynBD06}.
%Yannakakis81: Edge deletion problems (1) without cycles of specified length l, or of any length $ \leqq l$, (2) connected and degree-constrained, (3) outerplanar, (4) transitive digraph, (5) line-invertible, (6) bipartite, (7) transitively orientable
%NatanzonSS01: proper circular arc, unit circular arc, AT-free
%Colbourn88: Cographs, cluster
%Planar...

%%%%%%About FPT algorithm.
Parameterized complexity is well-studied in the area of computer science.
A problem with a parameter $k$ is \emph{fixed parameter tractable}, \emph{FPT} for short, if there is an algorithm running in $f(k)n^c$ time where $n$ is the size of input, $f$ is a computable function and $c$ is a constant.
Such an algorithm is called an \emph{FPT algorithm}.
The \emph{treewidth} $\tw(G)$ of a graph $G$ represents treelikeness and is one of the most important parameters in parameterized complexity concerning graph algorithms. 
For many NP-hard problems in general, there are tons of FPT algorithms with parameter $\tw(G)$ by dynamic programming on tree decompositions. 
Finding the treewidth of an input graph is NP-hard and it is known that \textsc{Chordal-Completion} with minimizing the size of the smallest maximum clique is equivalent to the problem. 
There is an FPT algorithm for computing the treewidth of a graph by Bodlaender~\cite{Bodlaender96} which  
runs in $O(f(\tw(G))(n+m))$ time where $n$ and $m$ are the numbers of vertices and edges of a given graph: i.e., the running time is linear in the size of input. % FPT.
Courcelle showed that every problem that can be expressed in monadic second order logic (MSO$_2$) has a linear time algorithm on graphs of bounded treewidth~\cite{CourcelleEngelfiet}.
%there is a linear time algorithm with the bounded treewidth of a graph if the problem is written by monadic second order logic (MSO$_2$)~\cite{CourcelleEngelfiet}. % courcelle の定理の簡単な説明．
Some intersection graph classes, for example interval graphs, proper interval graphs, chordal graphs, and permutation graphs, can be represented by MSO$_2$~\cite{Courcelle06} and thus there are FPT algorithms for \textsc{Edge-Deletion} problems on such graph classes. 
However, the algorithms obtained by Courcelle's theorem have a very large hidden constant factor even when the treewidth is very small, since the running time is the exponential tower of the coding size of the MSO$_2$ expression.

\medskip
\noindent{\bf Our results:}
%\subsubsection{Our results}
%We propose concrete FPT algorithms for \textsc{Edge-Deletion} to permutation graphs, interval graphs, and other related graph classes, when parameterized by the treewidth of the input graph.
We propose concrete FPT algorithms for \textsc{Edge-Deletion} to interval graphs and other related graph classes, when parameterized by the treewidth of the input graph.
%We propose concrete FPT algorithms for  \textsc{Edge-Deletion} to a number of intersection graph classes when parameterized by the treewidth of the input graph.
%, e.g. interval graphs and permutation graphs, etc.,
%We propose FPT algorithms for \textsc{Interval Edge Deletion}, \textsc{Permutation Edge Deletion}, and so on.
Our algorithms virtually compute a set of edges $S$ with the minimum size such that $G-S$ is in a graph class $\C$ by using dynamic programming on a tree-decomposition. 
We maintain possible alignments of geometric objects corresponding to vertices in the bag of each node of the tree-decomposition. % and the alignment of the geometric objects are represented as permutations. 
Alignments of the objects of forgotten vertices are remembered only relatively to the objects of the current bag.
If two forgotten objects have the same relative position to the objects of the current bag, we remember only the fact that there is at least one forgotten object at that position.   
In this way, we achieve the fixed-parameter-tractability, while guaranteeing that no object pairs of non-adjacent vertices of the input graph will intersect in our dynamic programming algorithm.
%However, all the permutations represented the geometric objects in the bag cannot be consistent with the other vertices since the objects of forgotten vertices on tree decomposition cannot intersect those of the introduced vertices. 
%To solve this problem, we give some restrictions to be consistent with the current bag. 
Our algorithms run in $O(f(\tw(G))\cdot (n+m))$ time where $n$ and $m$ are the numbers of vertices and edges of the input graph.
%and they are faster than Courcelle's algorithm.
Our explicit algorithms are significantly faster than those obtained by using Courcelle's theorem.
We also analyze the time complexity of our algorithms parameterized by pathwidth which is analogous to treewidth. 
%We show the results in Table~\ref{tab:results}.
%% {\bf CHECK RUNNING TIME: linear time for the graph size?}
The relation among the graph classes for which this paper provides $\C$-\textsc{Edge-Deletion} algorithms is shown in \figurename~\ref{fig:hasse}.
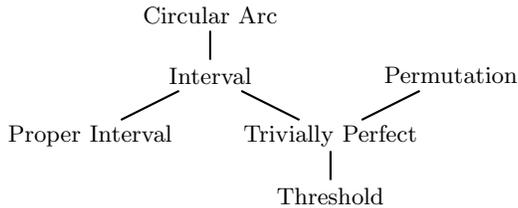
\begin{figure}[tbh]
	\centering
	\begin{tikzpicture}[scale=0.8, every node/.style=rectangle,fill=white,inner sep=0pt,minimum size=4mm]\small
		\node[] (CA) at (2,3) {Circular Arc};
		\node[] (IV) at (2,2) {Interval};
		\node[] (PI) at (0,1) {Proper Interval};
		\node[] (TP) at (4,1) {Trivially Perfect};
		\node[] (TS) at (4,0) {Threshold};
		\node[] (PM) at (6,2) {Permutation};
		\draw[thick,-] (CA) to (IV);
		\draw[thick,-] (IV) to (PI);
		\draw[thick,-] (IV) to (TP);
		\draw[thick,-] (TP) to (TS);
        \draw[thick,-] (TP) to (PM);
	\end{tikzpicture}
\iffalse
	\begin{tikzpicture}[scale=0.8, every node/.style=rectangle,fill=white,inner sep=0pt,minimum size=4mm]\small
		\node[] (CA) at (0,3) {Circular Arc};
		\node[] (IV) at (2,2) {Interval};
		\node[] (PI) at (0,1) {Proper Interval};
		\node[] (TP) at (4,1) {Trivially Perfect};
		\node[] (TS) at (4,0) {Threshold};
		\node[] (PM) at (5,2) {Permutation};
		\draw[thick,-] (CA) to (IV);
		\draw[thick,-] (IV) to (PI);
		\draw[thick,-] (IV) to (TP);
		\draw[thick,-] (TP) to (TS);
        \draw[thick,-] (TP) to (PM);
%
		\node[] (PCA) at (-1,2) {Proper Circular Arc};
		\node[] (CL) at (3,3) {Circle};
        \draw[thick,-] (CL) to (PM);
		\draw[thick,-] (CA) to (PCA);
        \draw[thick,-] (CL) to (PCA);
		\draw[thick,-] (PCA) to (PI);
	\end{tikzpicture}
\fi
\caption{\label{fig:hasse} The graph classes of which this paper presents algorithms for the edge-deletion problems.}
\end{figure}

%% \begin{table}
%%   \centering
%%   \caption{The time complexities of our algorithms. $k$ is the pathwidth or treewidth of an input graph and $\mrm{poly}$ is a polynomial function. }
%%   \label{tab:results}
%%   \begin{tabular}{c|cc}
%%     Graph class & Pathwidth & Treewidth\\\hline
%%     Permutation & $O((k!)^2\cdot2^{k^2}\mrm{poly}(k) \cdot |V|)$ & $O(((k!)^2\cdot2^{k^2})^2 \mrm{poly}(k) \cdot |V|)$\\
%%     Interval & $O((2k)!\cdot 2^{2k}\mrm{poly}(k) \cdot |V|)$ & $O(((2k)!\cdot 2^{2k})^2 \mrm{poly}(k) \cdot |V|)$\\
%%     Proper interval & $O((2k)!\cdot 2^{k}\mrm{poly}(k) \cdot |V|)$ & $O(((2k)!\cdot 2^{k})^2 \mrm{poly}(k) \cdot |V|)$\\
%%     Trivially perfect & $O((2k)!\cdot 2^{k}\mrm{poly}(k) \cdot |V|)$ & $O(((2k)!\cdot 2^{k})^2 \mrm{poly}(k) \cdot |V|)$\\
%%     Circular arc & $O((2k)!\cdot 2^{3k}\mrm{poly}(k) \cdot |V|)$ & $O(((2k)!\cdot 2^{3k})^2 \mrm{poly}(k) \cdot |V|)$\\
%%     %Chordal & & \\
%%     %% Chain & &\\
%%     %% Threshold graph & & \\
%%     %% Split (NO DISCUSSION)& & \\
%%   \end{tabular}
%% \end{table}

\medskip
\noindent{\bf Related works:}
%\subsubsection{Related work}
Another kind of common parameters considered in parameterized complexity of graph modification problems is the number of vertices or edges to be removed or to be added.
Here we review preceding studies on those problems for intersection graphs with those parameters.

%On a parameterized complexity of $\C$-\textsc{Vertex/Edge-Deletion}, $\C$-\textsc{Completion} problems,
%the parameter is the size of removing vertices/edges, or adding edges. 
Concerning parameterized complexity of $\C$-\textsc{Vertex-Deletion},
Hof et~al.\ proposed an FPT algorithm for \textsc{Proper-Interval-Vertex-Deletion}~\cite{HofV13}, and 
Marx proposed an FPT algorithm for \textsc{Chordal-Vertex-Deletion}~\cite{Marx10}.
Heggernes et~al.\ showed \textsc{Perfect-Vertex-Deletion} and \textsc{Weakly-Chordal-Vertex-Deletion} are W[2]-hard~\cite{HeggernesHJKV13}.
% we believe that there is no FPT algorithm. 
Cai showed that $C$-\textsc{Vertex/Edge-Deletion} are FPT when $\C$ is characterized by a finite set of forbidden induced subgraphs~\cite{Cai96}.
%F is caharacterized by a finite set of forbidden minors, then FPT.%要チェック???

For modification problems on interval graphs,
Villanger et~al.\ presented an FPT algorithm for \textsc{Interval-Completion}~\cite{VillangerHPT09},
and Cao and Marx presented an FPT algorithm for \textsc{Interval-Vertex-Deletion}~\cite{CaoM15}.
Cao improved these algorithms and developed an FPT algorithm for \textsc{Edge-Deletion}~\cite{Cao16}.
%Open question for interval edge deletion by Bodlaender et~al. %consider edge deletion, edge addition~\cite{Cao16}

%% kernelization. 
%% \textsc{Interval Vertex Deletion} admits a polynomial kernel~\cite{AgrawalM0Z19}
%% %Perfect deletion, Weakly chordal deletion%polynomial kernel は難しい
%% polynomial kernel for chordal deletion~\cite{HeggernesHJKV13}

It is known that \textsc{Threshold-Edge-Deletion}, \textsc{Chain-Edge-Deletion} and \textsc{Trivially-Perfect-Edge-Deletion} are FPT, since threshold graphs, chain graphs and trivially perfect graphs are characterized by a finite set of forbidden induced subgraphs~\cite{Cai96}. 
Nastos and Gao presented faster algorithms for the problems~\cite{NastosG12}, %$O(4^k(n+m))$ and $O(c^k(n+m))$
and Liu et~al.\ improved their algorithms to $O(2.57^k(n+m))$ and $O(2.42^k(n+m))$ using modular decomposition trees~\cite{LiuWYCC15}, where $k$ is the number of deleted edges.
There are algorithms to find a polynomial kernel for \textsc{Chain-Edge-Deletion} and \textsc{Trivially Perfect-Edge-Deletion}~\cite{BessyP13,DrangeP18}. %, when parameterized by the number of deleted edges. 
%These known FPT algorithms are parameterized by the size of deleted vertices or edges but our algorithms are parameterized by treewidth of a graph. 
%split graphs: characterized by {2K_2, C_4, C_5}-free, intersection graphs of star.

\medskip
\noindent{\bf Organization of this article:}
%\subsubsection{Organization of this article}
Section~2 prepares the notation and definitions used in this paper.  
We propose an FPT algorithm for \textsc{Interval-Edge-Deletion} in Section~3. 
%We propose an FPT algorithm for \textsc{Permutation-Edge-Deletion} in Section 3 and describe an FPT algorithm for \textsc{Interval-Edge-Deletion} in Section 4.
We then extend the algorithm related to the interval graphs in Section~4. %\textsc{Proper-Interval-Edge-Deletion}, \textsc{Trivially-Perfect-Edge-Deletion}, \textsc{Circular-Arc-Edge-Deletion}, and \textsc{Threshold-Edge-Deletion} in Sections 5.
%other graph classes, proper interval graphs, trivially perfect graphs, and circular-arc graphs in Sections 5, 6, and 7, respectively.
We conclude this paper and provide some open questions in Section~\ref{sec:conc}.
%Due to lack of the space, we provide a \textsc{Threshold-Edge-Deletion} algorithm in the appendix. 
%% The relation among those graph classes is shown in Figure~\ref{fig:hasse}.

%% \begin{figure}
%% %\centering
%% %\begin{subfigure}{\textwidth}
%% 	\centering
%% 	\begin{tikzpicture}[scale=0.8, every node/.style=rectangle,fill=white,inner sep=0pt,minimum size=4mm]\small
%% 		\node[] (CA) at (2,3) {Circular Arc};
%% 		\node[] (IV) at (2,2) {Interval};
%% 		\node[] (PI) at (0,1) {Proper Interval};
%% 		\node[] (TP) at (4,1) {Trivially Perfect};
%% 		\node[] (TS) at (4,0) {Threshold};
%% 		\node[] (PM) at (8,2) {Permutation};
%% 		\draw[thick,-] (CA) to (IV);
%% 		\draw[thick,-] (IV) to (PI);
%% 		\draw[thick,-] (IV) to (TP);
%% 		\draw[thick,-] (TP) to (TS);
%%                 \draw[thick,-] (TP) to (PM);
%% 	\end{tikzpicture}
%% %\end{subfigure}
%% \caption{\label{fig:hasse} The graph classes of which this paper presents algorithms for the edge-deletion problems.}
%% \end{figure}

% LocalWords:  FPT computable treewidth treelikeness decompositions MSO
% LocalWords:  Courcelle's pathwidth

%% file: preliminaries.tex
% !TEX root = main.tex
\section{Preliminaries}
For a set $X$, its cardinality is denoted by $|X|$. 
A \emph{partition} of $X$ is a tuple $(X_1,\dots,X_k)$ of subsets of $X$ such that $X=X_1 \cup \dots \cup X_k$ and $X_i \cap X_j = \emptyset$ if $1 \le i < j \le k$,
where we allow some of the subsets to be empty.
%Differently from the usual definition, 
For entities $x,y,z \in X$, we let $x[y/z] = y$ if $x = z$, and $x[y/z]=x$ otherwise.
For a subset $Y \subseteq X$, define $Y[y/z] = \{\, x[y/z] \mid x \in Y\,\}$.
%\[
%Y[x/y] = \begin{cases} Y & \text{ if $y \notin Y$,}
%\\	Y \cup \{x\} - \{y\} & \text{ if $y \in Y$.}
%\end{cases}
%\]

For a linear order $\pi$ over a finite set $X$, the maximum and the minimum elements of $X$ w.r.t.\ $\pi$ are denoted by $\max_\pi X$ and $\min_\pi X$, respectively.
We denote the successor of $x \in X$ w.r.t.\ $\pi$ by $\suc{\pi}(x) = \min_\pi\{\, y \in X \mid x <_\pi y\,\}$. % i.e., $x <_\pi \suc{\pi}(x)$ and $\suc{\pi}(x) \le_{\pi} y$ for all $y$ with $x <_{\pi} y$.
Note that $\suc{\pi}(\max_\pi X)$ is undefined.
%The maximum element of $X$ w.r.t.\ $\pi$ is denoted by $\max_\pi X$.
Similarly $\pred{\pi}(x)  = \max_\pi\{\, y \in X \mid y <_\pi x\,\}$ denotes the predecessor of $x$. % and $\min_\pi X$ is the least element of $X$.

% For two sets $X,Y$, we define $X \otimes Y = \{\,\{u,v\} \mid u \in X \text{ and } v \in Y\,\}$.
A simple graph $G=(V,E)$ is a pair of vertex and edge sets, where each element of $E$ is a subset of $V$ consisting of exactly two elements.
%If a vertex $u$ of an edge $\{u,v\}$ is from a specific subset $C \subseteq E$ and the other $v$ is from $D \subseteq E$, abusing the notation, we sometimes denote it as $(u,v) \in C \times D$.
%A \emph{path-decomposition} of $G$ is a sequence $(X_1,\dots,X_m)$ of sets of vertices such that
%\begin{itemize}
%	\item for each $u\in V$, there are $i,j$ such that $u \in X_k \iff i \le k \le j$,
%	\item for each $\{u,v\} \in E$, there is $k$ such that $u,v \in X_k$.
%\end{itemize}
%The \emph{width} of a path-decomposition  $(X_1,\dots,X_m)$ is $\max\{\, |X_k| \mid 1 \le k \le m \,\}$
%and the \emph{pathwidth} of a graph is the smallest width of its path-decompositions.
%A path-decomposition  $(X_1,\dots,X_m)$  is said to be \emph{canonical} if 
%\begin{itemize}
%	\item $X_1=X_m=\emptyset$,
%	\item the symmetric difference $X_i \triangle X_{i+1}$ of two conjunctive sets $X_i$ and $X_{i+1}$ is always a singleton.
%\end{itemize}
%It is known that every path-decomposition has a canonical path-decomposition of the same width.

A \emph{tree-decomposition} of $G=(V,E)$ is a tree $T$ such that\footnote{We use the terms ``vertices'' for an input graph and ``nodes'' for a tree-decomposition.}
\begin{itemize}
	\item to each node of $T$ a subset of $V$ is assigned,
	\item if the assigned sets of two nodes of $T$ contain a vertex $u \in V$, then so does every node on the path between the two nodes,
	\item for each $\{u,v\} \in E$, there is a node of $T$ whose assigned set includes both $u$ and $v$.
\end{itemize}
The \emph{width} of a tree-decomposition is the maximum cardinality of the assigned sets minus one
and the \emph{treewidth} of a graph is the smallest width of its tree-decompositions.
A tree-decomposition is said to be \emph{nice} if it is rooted, the root is assigned the empty set, and its nodes are grouped into the following four:
\begin{itemize}
	\item \emph{leaf nodes}, which have no children and are assigned the empty set,
	\item \emph{introduce nodes}, each of which has just one child, where the set assigned to the node adds one vertex to the child's set,
	\item \emph{forget nodes}, each of which has just one child, where the set assigned to the node removes one vertex from the child's set,
	\item \emph{join nodes}, each of which has just two children, where the same vertex set is assigned to the node and its children.
\end{itemize}
It is known that every tree-decomposition has a nice tree-decomposition of the same width whose size is $O(k|V|)$, where $k$ is the treewidth of the tree-decomposition~\cite{Kloks94,CyganFKLMPPS15}. %\ryoshi{Citation? T. Kloks ``Treewidth''?} in time $O(k^2\cdot |V|)$
Hereafter, under a fixed graph $G$ and a fixed nice tree-decomposition $T$, we let $X_s$ denote the subset of $V$ assigned to a node $s$ of a tree-decomposition and $X_{\le s}$ denote the union of all the subsets assigned to the node $s$ and its descendant nodes.
We call vertices in $X_s$ and in $X_{\le s} - X_s$ \emph{active} and \emph{forgotten}, respectively.
Moreover, we define $E_{s} = \{\, \{u,v\} \in E \mid u,v \in X_{s}\,\}$ and $E_{\le s} = \{\, \{u,v\} \in E \mid u,v \in X_{\le s}\,\}$. 
Given a tree decomposition of treewidth $k$, we can compute a nice tree-decomposition with treewidth $k$ and $O(kn)$ nodes in $O(k^2(n+m))$ time~\cite{CyganFKLMPPS15}. 

A tree-decomposition is called a \emph{path-decomposition} if the tree is a path.
The \emph{pathwidth} of a graph is the smallest width of its path-decompositions.
Every path-decomposition has a nice path-decomposition of the same pathwidth, which consists of leaf, introduce, forget, but not join nodes.

The problem we tackle in this paper is given as follows.
\begin{definition}
For a graph class $\C$, the \textsc{$\C$-Edge-Deletion} is a problem to find the minimum natural number $c$ such that
 there is a subgraph $G'=(V,E')$ of $G$ with $G' \in \C$ and $|E|-|E'| = c$ for an input simple graph $G = (V,E)$.
\end{definition}
In the succeeding sections,
for different classes $\C$ of intersection graphs,
 we present algorithms for \textsc{$\C$-Edge-Deletion} that run in linear time in the input graph size when the treewidth is bounded.
We assume that the algorithm takes a nice tree-decomposition $T$ of $G$ in addition as input~\cite{Bodlaender96,Kloks94}. %\ryoshi{Citation?}.
Our algorithms are dynamic programming algorithms that recursively compute solutions (and some auxiliary information) for the subproblems on $(X_{\le s},E_{\le s})$ for each node $s$ in the given tree-decomposition from leaves to the root.

% LocalWords:  cardinality treewidth pathwidth decompositions

%% file: interval.tex
% !TEX root = main_arxiv.tex
\section{Finding a Largest Interval Subgraph}\label{sec:interval}
An \emph{interval representation} $\pi$ over a set $X$ is a linear order over the set $\mathit{LR}_X = \mathit{L}_X \cup \mathit{R}_X \cup \{\bot,\top\}$ with $\mathit{L}_X = \{\, l_x \mid x \in X\,\}$ and $\mathit{R}_X = \{\, r_x \mid x \in X\,\}$ such that $\bot <_\pi l_x <_\pi r_x <_\pi \top$ for all $x \in X$.
An \emph{interval} is a pair $(p,q) \in \mathit{LR}_X \times \mathit{LR}_X$.
We say that two intervals $(p_1,q_1)$ and $(p_2,q_2)$ \emph{intersect}, denoted by $\its{\pi}{(p_1,q_1)}{(p_2,q_2)}$ if $p_1<_\pi q_2$ and $p_2 <_\pi q_1$.  Otherwise, we write $\prl{\pi}{(p_1,q_1)}{(p_2,q_2)}$.

The \emph{interval graph $G_\pi$ of an interval representation $\pi$ on $V$} is $(V,\mcal{E}_\pi)$ where
%$E_\pi = \{\,\{u,v\} \subseteq V \mid \text{$\itvl{l_u,r_u}_\pi \cap \itvl{l_v,r_v}_\pi \neq \emptyset$ and $u \neq v$} \,\}$.
\[
\mcal{E}_\pi = \{\,\{u,v\} \subseteq V \mid \text{$\its{\pi}{(l_u,r_u)}{(l_v,r_v)}$ and $u \neq v$} \,\}.
\]
\figurename~\ref{fig:interval} (a) and (b) show an example of an interval representation and the defined interval graph, respectively.
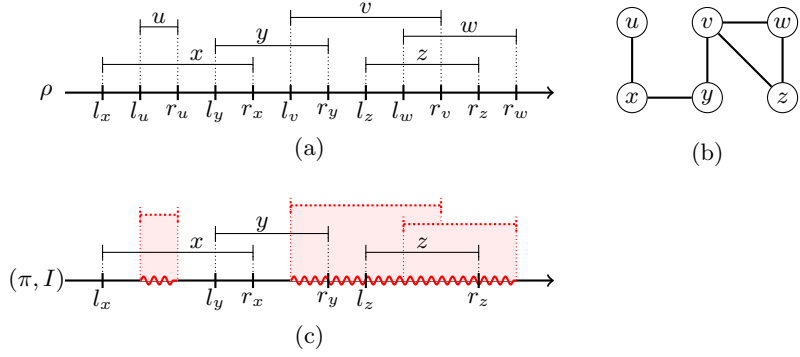
\begin{figure}[t]
%\centering
%\begin{subfigure}{\textwidth}
	\centering
	\begin{tikzpicture}[scale=0.5]\small
		\draw(5.5, -2.5) node{{(a)}};
		\draw(-1.5,-1) node{$\rho$};
		\draw[->,thick] (-1,-1)--(12,-1);	
		\rtdlritvl{0}{4}{-0.25}{-1}{$l_{x}$}{$r_{x}$}
		\draw (2.5,0) node{$x$};
		\rtdlritvl{3}{6}{0.25}{-1}{$l_{y}$}{$r_{y}$}
		\draw (4.25,0.5) node{$y$};
		\rtdlritvl{7}{10}{-0.25}{-1}{$l_{z}$}{$r_{z}$}
		\draw (8.5,0) node{$z$};
		\rtdlritvl{1}{2}{0.75}{-1}{$l_{u}$}{$r_{u}$}
		\draw (1.5,1) node{$u$};
		\rtdlritvl{5}{9}{1}{-1}{$l_{v}$}{$r_{v}$}
		\draw (7,1.25) node{$v$};
		\rtdlritvl{8}{11}{0.5}{-1}{$l_{w}$}{$r_{w}$}
		\draw (9.75,0.75) node{$w$};
%%%%%%%%%%%%
%%%%%%%%%%%%
%%%%%%%%%%%%
	\begin{scope}[shift={(0,-5)}]
		\draw(5.5, -2.5) node{{(c)}};
		\draw(-1.75,-1) node{$(\pi,I)$};
		\draw[->,thick] (-1,-1)--(12,-1);	
		\rtdlritvlf{1}{2}{0.75}{-1}{}{}
		\rtdlritvlf{5}{9}{1}{-1}{}{}
		\rtdlritvlf{8}{11}{0.5}{-1}{}{}
		\interval[-, thick, snake=snake, segment amplitude=1.5pt, segment length=4pt, color=red]{1}{2}{-1}
		\interval[-, thick, snake=snake, segment amplitude=1.5pt, segment length=4pt, color=red]{5}{11}{-1}
		\rtdlritvl{0}{4}{-0.25}{-1}{$l_{x}$}{$r_{x}$}
		\draw (2.5,0) node{$x$};
		\rtdlritvl{3}{6}{0.25}{-1}{$l_{y}$}{$r_{y}$}
		\draw (4.25,0.5) node{$y$};
		\rtdlritvl{7}{10}{-0.25}{-1}{$l_{z}$}{$r_{z}$}
		\draw (8.5,0) node{$z$};
	\end{scope}
	\end{tikzpicture}
\quad\quad
	\begin{tikzpicture}[scale=0.5, every node/.style=circle,fill=white,inner sep=0pt,minimum size=4mm]\small
	\node (o) at (0,-6) {}; %dummy
	\draw(2, -1.0) node{{(b)}};
	\node[draw,fill] (u1) at (0,0.5) {$x$};
	\node[draw,fill] (u2) at (2,0.5) {$y$};
	\node[draw,fill] (u3) at (4,0.5) {$z$};
	\node[draw,fill] (w1) at (0,2.5) {$u$};
	\node[draw,fill] (w2) at (2,2.5) {$v$};
	\node[draw,fill] (w3) at (4,2.5) {$w$};
	\draw[thick,-] (u1) to (u2);
	\draw[thick,-] (u1) to (w1);
	\draw[thick,-] (u2) to (w2);
	\draw[thick,-] (u3) to (w2);
	\draw[thick,-] (u3) to (w3);
	\draw[thick,-] (w2) to (w3);
	\end{tikzpicture}
%\end{subfigure}
\caption{\label{fig:interval} (a) Visualization of interval representation $\rho$. (b) Interval graph $G_\rho$.
(c) For $X_s = \{x,y,z\}$ and $X_{\le s} - X_s = \{u,v,w\}$, we have $[u]_\rho^s = \{l_u,r_u\}$, $[v]_\rho^s = [w]_\rho^s = \{l_v,r_v,l_w,r_w\}$, and $\mscr{A}(\rho,s) = (\pi,I,0)$ where $I = \{(l_{x},l_{x}),(r_{x},r_{z})\}$.}
\end{figure}

This section presents an FPT algorithm for the interval edge deletion problem w.r.t.\ the treewidth.
Let $G = (V, E)$ be an input graph and $s$ a node of a nice tree-decomposition $T$ of $G$.
On each node $s$ of $T$, for each interval representation $\rho$ over $X_{\le s}$ that gives an interval subgraph of $(X_{\le s},E_{\le s})$, we would like to remember some pieces of information about $\rho$, which we call the ``abstraction'' of $\rho$.
The abstraction is the triple $(\pi,I,c)$,
where the linear order $\pi$ over $\mathit{LR}_{X_s}$ is the restriction of $\rho$ to $X_s$,
forgotten vertices are represented in $I$ by anchoring the ends of intervals formed by forgotten vertices to active vertices,
and $c$ counts the number of edges of $E-E_\rho$ such that at least one of their ends is forgotten.
For each forgotten vertex $w \in X_{\le s} - X_s$, let the \emph{intersection closure of $w$ (w.r.t.\ $\rho$ and $s$)} be the smallest set $[w]_\rho^s \subseteq \mathit{LR}_{X_{\le s} - X_s}$ such that
\begin{itemize}
\item $l_w,r_w \in [w]_\rho^s$ and
\item if $l_{u},r_{u} \in [w]_\rho^s$, $v \in X_{\le s} - X_s$ and $\its{\rho}{(l_{u},r_{u})}{(l_{v},r_{v})}$, then $l_{v},r_{v}  \in [w]_\rho^s$.
\end{itemize}
%We usually omit the superscript $s$ and the subscript $\rho$ if they are understood from the context.
For an interval representation $\rho$ over $X_{\le s}$ such that $G_\rho$ is a subgraph of $(X_{\le s}, E_{\le s})$,
 we define the \emph{abstraction $\mscr{A}(\rho,s)$ of $\rho$ for $s$} to be the triple $(\pi,I,c)$ such that%
\footnote{Actually presence of intervals $(\bot,\bot)$ and $(\pred{\pi}(\top),\pred{\pi}(\top))$ in $I$ does not matter.}
\begin{itemize}
	\item $\pi$ is the restriction of $\rho$ to $\mathit{LR}_{X_s}$,
	\item $I = \{\, (p,q) \in \mathit{LR}_{X_s} \times \mathit{LR}_{X_s} \mid \text{there is $w \in X_{\le s} - X_s$ such that}$\\ \text{\hspace{2.5cm}$p <_{\rho} \min_{\rho}[w]_\rho^s <_{\rho} \suc{\pi}(p)$ and $q <_{\rho} \max_{\rho}[w]_\rho^s <_{\rho} \suc{\pi}(q)$}\,\},
	\item $c = |E_{\le s} - \mcal{E}_\rho - E_s| $ \\ \phantom{$c$}${}
	= |\{\, \{u,v\} \in E_{\le s} \mid u \notin {X_s}$ and $\prl{\rho}{(l_u,r_u)}{(l_v,r_v)}\,\}|
	$.
\end{itemize}
We call elements of $I$ \emph{forbidden intervals} because introducing a new vertex interval intersecting a forbidden interval means making the new vertex and a forgotten vertex adjacent, which is forbidden.
Figure~\ref{fig:interval} (c) shows an example of $I$.
Since forbidden intervals are formed by intersection closures, no distinct forbidden intervals may intersect.
%\begin{lemma}
%	For any $(p_1,q_1),(p_2,q_2) \in I$,  $\prl{\pi}{(p_1,q_1)}{(p_2,q_2)}$ unless $(p_1,q_1)=(p_2,q_2)$.
%\end{lemma}
If $s$ is the root of a nice tree-decomposition, the abstraction will be $\mscr{A}(\rho,s)=(o,\{(\bot,\bot)\},c_\rho)$ where $o$ is the trivial order on $\{\bot,\top\}$ such that $\bot <_o \top$ and $c_\rho = |E-E_\rho|$. That is, the smallest $c_\rho$ for an interval subgraph $G_\rho$ is the solution to our problem.
%\begin{proof}
%Suppose $(p_1,q_1),(p_2,q_2) \in I$ and $\its{\pi}{(p_1,q_1)}{(p_2,q_2)}$.
%% Suppose $\its{\pi}{(p_1,q_1)}{(p_2,q_2)}$, i.e., $p_1 <_\pi q_2$ and $p_2 <_\pi q_1$.
%There are forgotten vertices $u_1,v_1,u_2,v_2 \in X_{\le s} - X_s$ such that
% $p_i <_{\rho} l_{u_i} <_{\rho} \suc{\pi}(p_i)$,
% $q_i <_{\rho} r_{v_i} <_{\rho} \suc{\pi}(q_i)$, and 
% $[u_i]=[v_i]$ for each $i \in \{1,2\}$.
%It is enough to show $[u_1]=[u_2]$ by induction on $|[u_1]|+|[u_2]|$.
%If $|[u_1]|=|[u_2]|=2$, then $u_1=v_1$ and $u_2=v_2$.
%In this case, $l_{u_1} <_{\rho} \suc{\pi}(p_1) \le_{\pi} q_2 <_{\rho} r_{u_2}$ and $l_{u_2} <_{\rho} \suc{\pi}(p_2) \le_{\pi} q_1 <_{\rho} r_{u_1}$ hold.
%That is, $\its{\rho}{(l_{u_1},r_{u_1})}{(l_{u_2},r_{u_2})}$ and thus $[u_1]=[u_2]$.
%If $|[u_1]|+|[u_2]| > 4$, without loss of generality we may assume $|[u_2]| > 2$.
%
%
%\end{proof}

However, we do not have to compute $\mscr{A}(\rho,s)$ of all the possible interval representations $\rho$ over $X_{\le s}$ on each node $s$.
We say that $(\pi,I,c)$ \emph{dominates} $(\pi',I',c')$ if and only if $\pi=\pi'$, $I \sqsubseteq_\pi I'$ and $c \le c'$,
where we write $I \sqsubseteq_\pi I'$ if every forbidden interval of $I$ is inside some forbidden interval of $I'$:
i.e., for every $(p,q) \in I$, there is $(p',q') \in I'$ such that $p' \le_{\pi} p$ and $q \le_{\pi} q'$.
In this case, every possible way of introducing new intervals to $\rho'$ is also possible for $\rho$ by cheaper or equivalent cost.
Therefore, it is enough to remember $\mscr{A}(\rho,s)$ discarding $\mscr{A}(\rho',s)$.
For a set of abstractions, the process of removing ones which are dominated by others is called \emph{reduction}.
We call a set of abstractions \emph{reduced} if it has no pair of distinct elements such that one dominates the other.
\begin{lemma}
Suppose $\mscr{A}(\rho,s)=(\pi,I,c)$ dominates $\mscr{A}(\rho',s)=(\pi,I',c')$.
For any extension $\sigma'$ of $\rho'$ such that $\mcal{E}_{\sigma'} \subseteq E_{\le s}$, there is an extension $\sigma$ of $\rho$ such that $\mcal{E}_{\sigma'} \subseteq \mcal{E}_{\sigma} \subseteq E_{\le s}$.
\end{lemma}

%One can effectively ``reduce'' a set of abstractions, i.e., by removing elements dominated by another, one can get the maximum reduced subset.

%A naive idea is to remember all intervals of the forgotten vertices $w \in X_{\le s}-X_s$ relative to those of the active vertices in $X_s$ under $\rho$.
%That is, we remember the set of pairs $(p,q) \in \mathit{LR}_X\times\mathit{LR}_X$ for which there exists $w \in X_{\le s}-X_s$ such that
%\(
% p <_{\rho} l_w <_{\rho} \suc{\pi}(p)
%\) and \(
% q <_{\rho} r_w <_{\rho} \suc{\pi}(q)
%%q = \max\{\,q \in X_s \mid q <_{\rho} r_w\,\}
%\), where $\pi$ is the restriction of $\rho$ to $\mathit{LR}_{X_s}$.
%Using those pairs, we can recall around where the forgotten intervals have been put, so that we know where we can put new intervals without intersecting forgotten ones.
%There can be at most $|\mathit{LR}_{X_s}|^2 \leq 2k^2+o(k^2)$ pairs to remember for each $\rho$ with regardless of the number of forgotten vertices where $k$ is the treewidth of $G$.
%One can design an FPT algorithm for the interval edge deletion problem based on this idea.
%However, the FPT algorithm we will present below is more efficient.
%The set of point pairs $(p,q) \in \mathit{LR}_X\times\mathit{LR}_X$ will be replaced with three sets $I,J,K \subseteq \mathit{LR}_X$ of single points. 

%\subsection{Algorithm invariant}
%\noindent {\bf Algorithm invariant: }
%Let $G=(V,E)$ be an input graph and $s$ be a node of a nice tree-decomposition $T$ of $G$.

Our algorithm calculates a reduced set $\mscr{I}_s$ of abstractions of interval representations of interval subgraphs of $({X}_{\le s},{E}_{\le s})$ for each node $s$ of $T$ which satisfies the following invariant.
\begin{condition}\label{cond:invariant_i}{\ }
\begin{itemize}
 \item Every element $(\pi,I,c) \in \mscr{I}_s$ is the abstraction of some interval representation of an interval subgraph of $(X_{\le s}, E_{\le s})$ for $X_s$,
 \item Any interval representation $\rho$ of any interval subgraph of $(X_{\le s}, E_{\le s})$ has an element of $\mscr{I}_s$ that dominates its abstraction $\mscr{A}(\rho,X_s)$.
 \end{itemize}
\end{condition}
%Clearly if $\mscr{I}_s$ satisfies the above condition if and only if so is its reduced form.
Since $\mscr{I}_s$ is reduced, if $X_s=\emptyset$, we have $\mscr{I}_s=\{(o,I,c)\}$ for some $I \subseteq \{(\bot,\bot)\}$ and $c \in \mbb{N}$, where $o$ is the trivial order such that $\bot <_o \top$.
Particularly for the root node $s$, the number $c$ is the least number such that one can obtain an interval subgraph by removing $c$ edges from $G$.
That is, $c$ is the solution to our problem.
If $s$ is a leaf, $\mscr{I}_s=\{(o,\emptyset,0)\}$ by definition.
It remains to show how to calculate $\mscr{I}_{s}$ from the child(ren) of $s$, while preserving the invariant (Condition~\ref{cond:invariant_i}).

%\subsection{Algorithm}

%\paragraph*{Leaf Node} If $s$ is a leaf, let  $\mscr{I}_s=\{(o,\emptyset,\emptyset,\emptyset,0)\}$.

%\noindent \textbf{\textsf{Introduce Node:}}
\paragraph*{Introduce Node:}
Suppose that $s$ is an introduce node. It has just one child $t$ such that $X_s = X_t \cup \{x\}$.
For an extension $\pi'$ of $\pi$ for $(\pi,I,c) \in \mscr{I}_t$, let us anchor the new points $l_x$ and $r_x$ to points $p_0$ and $q_0$ in $\mathit{LR}_{X_t}$:
\begin{align*}
	p_0 %&= \max_{\pi}\{\,p\in \mathit{LR}_{X_t} \mid p<_{\pi'} l_x \,\}
		&= \pred{\pi'}(l_x) \,,
\\
	q_0 %&= \max_{\pi}\{\,p\in \mathit{LR}_{X_t} \mid p <_{\pi'} r_x \,\}
		&= \begin{cases}
			\pred{\pi'}(r_x)	&	\text{if $\pred{\pi'}(r_x) \neq l_x$,}
\\			\pred{\pi'}(l_x) 	&	\text{otherwise.}
\end{cases}	
%	p_0 = \max_{\pi}\{\,p\in X_t \mid p<_{\pi'}l_x \,\}\,,
%\\
%	q_0 = \min_{\pi}\{\,p\in X_t \mid r_x <_{\pi'} p \,\}\,.
\end{align*}
We say $\pi'$ \emph{respects} $E$ and $I$ if
\begin{itemize}
		\item $\{ x,u \} \notin E$ for $u \in X_{t}$ implies $\prl{\pi'}{(l_x,r_x)}{(l_u,r_u)}$,
		\item $(p,q) \in I$ implies $\prl{\pi}{(p_0,q_0)}{(p,q)}$,
\end{itemize}
respectively.
If $\pi'$ does not respect $E$ (\emph{resp.}\ $I$), then it means that we are creating an edge between $x$ and a vertex $u$ in $X_s$ (\emph{resp.}\ in $X_{\le s} - X_s$) which are not adjacent in the input graph $G$.
Figure~\ref{fig:interval_introduce} shows examples of extensions that do or do not respect $I$. 
For each interval representation $\pi'$ extending $\pi$ to $\mathit{LR}_{X_s}$ that respects $E$ and $I$, we put one or two elements into $\mscr{I}_s'$ by the following manner.
Some forbidden intervals that had an anchor $p_0$ or $q_0$ in $\mscr{I}_{t}$ may be re-anchored to $l_x$ or $r_x$.
Forbidden intervals are partitioned into three, where the ones in $I_\mrm{L}$ and $I_\mrm{R}$ are certainly left to and right to $x$, respectively: 
\begin{align*}
	I_\mrm{L} &= \{\, (p,q) \in I \mid p <_{\pi} q_0 \,\} \,,
\\	I_\mrm{R} &= \{\, (p,q) \in I \mid p_0 <_{\pi} q \,\} \,,
\\	I_\mrm{M} &= \{\, (p,q) \in I \mid p = q = p_0 = q_0 \,\} \,.% I-I_\mrm{L}-I_\mrm{R}\,.
\end{align*}
Note that $I_\mrm{M}$ is either empty or a singleton and that $I_\mrm{L}$ and $I_\mrm{R}$ are mutually exclusive, since $\pi'$ respects $I$.
If a forbidden interval has a point anchored to $q_0$ and is right to $x$, then it will be re-anchored:
\begin{align*}
	I_\mrm{R}' &= \{\, (p[r_x/q_0],q[r_x/q_0]) \mid (p,q) \in I_\mrm{R} \,\} \,,
\\	I_\mrm{M}' &= \{\, (r_x,r_x) \mid (q_0, q_0) \in I_\mrm{M} \,\} \,.
\end{align*}
Corresponding to the possible choices, we put both $(\pi',I_\mrm{L} \cup I_\mrm{M} \cup I_\mrm{R}',c)$ and $(\pi',I_\mrm{L} \cup I_\mrm{M}' \cup I_\mrm{R}',c)$ into $\mscr{I}'_s$.
We will not add $(\pi',I_\mrm{L} \cup I_\mrm{M} \cup I_\mrm{M}' \cup I_\mrm{R}',c)$ since it is dominated by the other two.
\begin{figure}
	\renewcommand{\rtdlritvl}[7][]{%
	\lrinterval[#1]{#2}{#3}{#4}
	\draw[#1, densely dotted] (#2,#4) -- (#2,#5);
	\draw[#1, densely dotted] (#3,#4) -- (#3,#5);
	\point[#1]{#2}{#5}
	\point[#1]{#3}{#5}
	\draw (#2,#5-0.5) node{#6};
	\draw (#3,#5-0.5) node{#7};
}
	\centering
	\begin{tikzpicture}[scale=0.5]\small
		\draw(-0.5,0) node{$(\pi,I)$};
		\draw[->,thick] (0.5,0)--(13.5,0);
		\point[thick]{1}{0}
		\draw (1,-0.5) node{$m$};
		\point[thick]{3}{0}
		\draw (3,-0.5) node{$n$};
		\point[thick]{5}{0}
		\draw (5,-0.5) node{$p$};
		\point[thick]{11.5}{0}
		\draw (11.5,-0.5) node{$q$};
		\interval[finterval]{2}{4}{0}
		\interval[finterval]{7}{8.5}{0}
		\interval[finterval]{10.5}{12.5}{0}
		\rtdlritvl[color=blue]{5.5}{6.5}{2}{0}{}{}
		\draw (6,2.25) node[color=blue]{$x_1$};
		\rtdlritvl[color=blue]{9}{10}{2}{0}{}{}
		\draw (9.5,2.25) node[color=blue]{$x_2$};
		\rtdlritvl[color=blue]{4.5}{6.4}{1.3}{0}{}{}
		\draw (5.2,1.55) node[color=blue]{$x_3$};
		\rtdlritvl[color=blue]{1.5}{6.3}{0.6}{0}{}{}
		\draw (3.75,0.85) node[color=blue]{$x_4$};
		\rtdlritvl[color=blue]{9.1}{13}{1}{0}{}{}
		\draw (11,1.25) node[color=blue]{$x_5$};
	\end{tikzpicture}
\caption{\label{fig:interval_introduce}%
Introduce Node.
When $I = \{(m,n),(p,p),(p,q)\}$, new intervals like $x_1,x_2,x_3$ respect $I$, while $x_4,x_5$ do not.
If $\pi'$ introduces a new interval $(l_x,r_x)$ so that $p <_{\pi'} l_x <_{\pi'} r_x <_{\pi'} q$, one can assume that $(l_x,r_x)$ is either left to or right to the forbidden interval anchored to $(p,p)$. Those two possibilities are illustrated as $x_1$ and $x_2$. The respective cases give $\{(m,n),(r_x,r_x),(r_x,q)\}$ and $\{(m,n),(p,p),(r_x,q)\}$ as new forbidden interval sets.
}
\end{figure}
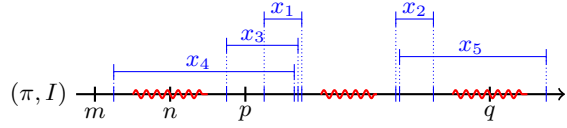
\begin{example}
Figure~\ref{fig:interval_introduce} shows several possibilities of extending $\pi$ ($m <_{\pi} n <_{\pi} p <_{\pi} q$) to $\pi'$ by introducing $l_x$ and $r_x$.
Defining $\pi'$ by letting $m <_{\pi'} n <_{\pi'} l_x <_{\pi'} p <_{\pi'} r_x <_{\pi'} q$ is illustrated by the interval of $x_3$,
where $I_\mrm{L}=\{(m,n) \}$, $I_\mrm{R}=\{(p,p),(p,q) \}$, and  $I_\mrm{M}=\varnothing$, in which case the anchors $(p,p),(p,q) \in I_\mrm{R}$ of forbidden intervals will be updated to $(r_x,r_x),(r_x,q)$.
Defining $\pi'$ by letting $m <_{\pi'} n <_{\pi'} p <_{\pi'} l_x <_{\pi'} r_x <_{\pi'} q$ gives $I_\mrm{L}=\{(m,n) \}$, $I_\mrm{R}=\{(p,q) \}$, and  $I_\mrm{M}=\{(p,p)\}$.
We consider two ways to put the interval $(l_x,r_x) $ relative to the forbidden interval anchored to $(p,p) \in I_\mrm{M}$.
One is to put $(l_x,r_x)$ left to those forbidden intervals, like the interval of $x_1$ shows, in which case we update $(p,p) \in I_\mrm{M}$ to $(r_x,r_x)$, so we obtain $I_1 = I_\mrm{L} \cup I_\mrm{M}' \cup I_\mrm{R}' = \{(m,n),(r_x,r_x),(r_x,q)\}$.
The other is to put $(l_x,r_x)$ right to them, as $x_2$ shows, in which case we keep $(p,p) \in I_\mrm{M}$, so we obtain $I_2 = I_\mrm{L} \cup I_\mrm{M} \cup I_\mrm{R}'=\{(m,n),(p,p),(r_x,q)\}$.
We note that there can be several non-intersecting forbidden intervals between $p$ and $\suc{\pi}(p)$, and one may locate $(l_x,r_x)$ between them.
This interpretation gives $I'=\{(m,n),(p,p),(r_x,r_x),(r_x,q)\}$ but we ignore this possibility, due to $I_1,I_2 \sqsubseteq I'$.
\end{example}
We then obtain $\mscr{I}_{s}$ by reducing $\mscr{I}'_{s}$.
%\begin{lemma}
%	Let $s$ be an introduce node and $t$ its child.
%	If an interval subgraph of $G_{\le s}$ has an interval representation $\rho$, then $\pi'$ respects $E$ and $I$ where $\pi'$ is the restriction to $X_s$ and $\mscr{A}(\rho,t)=(\pi,I,c)$.
%\end{lemma}
%\begin{proof}
%	It is obvious that $\pi'$ respects $E$.
%	For each forgotten vertex $u \in X_{\le s} - X_s$, since $\prl{\rho}{(l_x,r_x)}{(l_u,r_u)}$, either $r_x <_\rho l_u$ or $r_u <_\rho l_x$ holds.
%	So those forgotten vertices $X_{\le s} - X_s$ are partitioned into two: $X_\mrm{L} = \{\, u \mid r_u <_\rho l_x\,\}$ and $X_\mrm{R} = \{\, u \mid r_x <_\rho l_u\,\}$.
%	Therefore, $[u]_\rho^s \neq [v]_\rho^s$ for any $u \in X_\mrm{L}$ and $v \in X_\mrm{R}$.
%	Let $I_\mrm{L}$ and $I_\mrm{R}$ be the partition of $I$ where $I_\mrm{L}$ and $I_\mrm{R}$ originate in $X_\mrm{L}$ and  $X_\mrm{R}$, respectively,
%	$p_0 = \pred{\pi'}(l_x)$, and $q_0 = \max_{\pi}\{\,p\in \mathit{LR}_{X_t} \mid p <_{\pi'} r_x \,\}$.
%	Then, $(p,q) \in I_\mrm{L}$ implies $q \le_\pi p_0$ and $\prl{\pi}{(p_0,q_0)}{(p,q)}$.
%	Similarly,  $(p,q) \in I_\mrm{R}$ implies $q_0 \le_\pi p$ and $\prl{\pi}{(p_0,q_0)}{(p,q)}$.
%\end{proof}
%\begin{lemma}
%	If $\pi'$ respects $E$ and $I$, there is an interval subgraph of $G_{\le s}$ that has an interval representation $\rho$ which extends $\pi'$.
%\end{lemma}
%\begin{proof}
%	We show the lemma assuming that the invariant is true for the child $t$ of $s$.
%	Let $\rho$ be an extension of $\pi$ for which $\mcal{E}_\rho \subseteq E_{\le s}$.
%	We define $\rho'$ extending $\rho$ in two ways:
%\end{proof}

\paragraph*{Forget Node:}
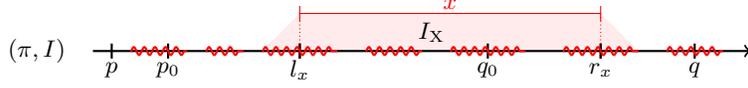
\begin{figure*}[t]\centering
%	\begin{tikzpicture}[scale=0.5]\small
%		\draw(-0.5,0) node{$(\pi,I)$};
%		%
%		\fill[red!8] (2,0) -- (4,1) -- (12,1) -- (13,0);
%		\rtdlritvl[color=red]{4}{12}{1}{0}{$l_x$}{$r_x$}
%		\draw (7.5,0.5) node{$I_\mrm{X}$};
%		%
%		\draw[->,thick] (0.5,0)--(16,0);
%		\point[thick]{1}{0}
%		\draw (1,-0.5) node{$p$};
%		\point[thick]{3}{0}
%		\draw (3,-0.5) node{$p_0$};
%		\point[thick]{4}{0}
%		\point[thick]{9}{0}
%		\draw (9,-0.5) node{$q_0$};
%		\point[thick]{12}{0}
%		\point[thick]{14.5}{0}
%		\draw (14.5,-0.5) node{$q$};
%		%
%		\interval[finterval]{2}{5}{0}
%		\interval[finterval]{5.75}{7.25}{0}
%		\interval[finterval]{8}{10}{0}
%		\interval[finterval]{11}{13}{0}
%		\interval[finterval]{13.75}{15.25}{0}
%		%
%		\draw (8,1.25) node[color=red]{$x$};
%	\end{tikzpicture}
	\begin{tikzpicture}[scale=0.5]\small
		\draw(-3,0) node{$(\pi,I)$};
		\fill[red!8] (3,0) -- (4,1) -- (12,1) -- (13,0);
		\rtdlritvl[color=red]{4}{12}{1}{0}{$l_x$}{$r_x$}
		\draw (7.5,0.5) node{$I_\mrm{X}$};
		\draw[->,thick] (-1.5,0)--(16,0);
		\point[thick]{-1}{0}
		\draw (-1,-0.5) node{$p$};
		\point[thick]{0.5}{0}
		\draw (0.5,-0.5) node{$p_0$};
		\point[thick]{4}{0}
		\point[thick]{9}{0}
		\draw (9,-0.5) node{$q_0$};
		\point[thick]{12}{0}
		\point[thick]{14.5}{0}
		\draw (14.5,-0.5) node{$q$};
		\interval[finterval]{-0.5}{1}{0}
		\interval[finterval]{1.5}{2.5}{0}
		\interval[finterval]{3}{5}{0}
		\interval[finterval]{5.75}{7.25}{0}
		\interval[finterval]{8}{10}{0}
		\interval[finterval]{11}{13}{0}
		\interval[finterval]{13.75}{15.25}{0}
		\draw (8,1.25) node[color=red]{$x$};
	\end{tikzpicture}
\caption{\label{fig:intersection_forget}%
Forget Node.
%$I=\{(p,l_x),(l_x,l_x),(l_x,q_0),(q_0,r_x),(r_x,q)\}$ will become $I'=\{(p,q_0),(q_0,q)\}$
$I=\{(p,p_0),(p_0,p_0),(p_0,l_x),(l_x,l_x),(l_x,q_0),(q_0,r_x),(r_x,q)\}$ will become $I'=\{(p,p_0),(p_0,p_0),(p_0,q_0),(q_0,q)\}$ after $x$ has been forgotten.}
\end{figure*}
Suppose that $s$ is a forget node. It has just one child $t$ such that $X_t = X_s \cup \{x\}$.
For each $(\pi,I,c)$ in $\mscr{I}_t$, in accordance with the definition of abstractions, we add to $\mscr{I}'_{s}$ the triple $(\pi',I',c)$ where
\begin{itemize}
	\item $\pi'$ is the restriction of $\pi$,
	\item $c' = c + |\{\,\{x,u\} \in E \mid \prl{\pi}{(l_u,r_u)}{(l_x,r_x)}$ and $u \in X_{s} \,\}|$.
\end{itemize}
We make $I'$ from $I$ by
\begin{itemize}
\item making a new forbidden interval involving $(l_x,r_x)$ and
\item re-ahcoring forbidden intervals if they have an anchor $l_x$ or $r_x$.
\end{itemize}
Let us anchor the points $l_x$ and $r_x$ to points $p_0$ and $q_0$ in $X_s$:
\begin{align*}
	p_0 &= \max_{\pi'}\{\,p\in X_s \mid p<_{\pi} l_x \,\}
		= \pred{\pi}(l_x) \,,
\\
	q_0 &= \max_{\pi'}\{\,p\in X_s \mid p <_{\pi} r_x \,\}
		= \begin{cases}
			\pred{\pi}(r_x)	&	\text{if $\pred{\pi}(r_x) \neq l_x$,}
\\			\pred{\pi}(l_x) 	&	\text{otherwise.}
\end{cases}	
\end{align*}
The set $I_X \subseteq I$ of forbidden intervals that intersect with $(l_x,r_x)$ will be given below.
They will be merged into one in $I'$.
\begin{align*}
	I_\mrm{X} &= \{\, (p,q) \in I \mid \its{\pi}{(p_0,r_x)}{(p,q)} \,\}
\\	p_* &= \min_\pi(\{p_0\} \cup \{\, p \mid (p,q) \in I_\mrm{X} \,\})
\\	q_* &= \max_\pi(\{r_x\} \cup \{\, q \mid (p,q) \in I_\mrm{X} \,\})
\end{align*}
Figure~\ref{fig:intersection_forget} may help understanding why we take $(p,q) \in I$ ``intersecting'' with $(p_0,r_x)$ rather than with $(l_x,r_x)$.
This comes from the gap of the meanings of position pairs $(l_x,r_x)$ and $(p,q) \in I$.
While the interval $(l_x,r_x)$ of $x$ starts exactly at $l_x$ and ends at $r_x$, the forbidden interval anchored to $(p,q) \in I$ starts after $p$ and ends after $q$.
Although forbidden intervals anchored to $(p,l_x)$ for some $p \le_\pi l_x$ must intersect with the interval $(l_x,r_x)$ of $x$, we have $\prl{\pi}{(p,l_x)}{(l_x,r_x)}$ according to the definition.
On the other hand, forbidden intervals anchored to $(r_x,q)$ for some $q \ge_\pi r_x$ do not intersect with $(l_x,r_x)$ and $\prl{\pi}{(r_x,q)}{(l_x,r_x)}$ holds.
The new forbidden interval will be 
\begin{align*}
	I' &= ( I-I_\mrm{X} \cup \{(p_*,q_*)\})[q_0/r_x]\,.
\end{align*}
Then we obtain $\mscr{I}_{s}$ by reducing $\mscr{I}'_{s}$.

\paragraph*{Join Node:}
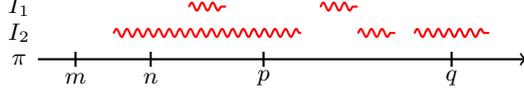
\begin{figure*}[t]\centering
	\begin{tikzpicture}[scale=0.5]\small
		\draw(-2.5,0) node{$\pi$};
		\draw(-2.5,1.4) node{$I_1$};
		\draw(-2.5,0.7) node{$I_2$};
		\draw[->,thick] (-2,0)--(11,0);
		\point[thick]{-1}{0}
		\draw (-1,-0.5) node{$m$};
		\point[thick]{1}{0}
		\draw (1,-0.5) node{$n$};
		\point[thick]{4}{0}
		\draw (4,-0.5) node{$p$};
		\point[thick]{9}{0}
		\draw (9,-0.5) node{$q$};
		\interval[finterval]{2}{3}{1.4}
		\interval[finterval]{5.5}{6.5}{1.4}
		\interval[finterval]{0}{5}{0.7}
		\interval[finterval]{6.5}{7.5}{0.7}
		\interval[finterval]{8}{10}{0.7}
	\end{tikzpicture}
\caption{\label{fig:interval_join}%
Join Node.
$A_1=(\pi,I_1,c_1)$ and $A_2=(\pi,I_2,c_2)$ with $I_1=\{(n,n),(p,p)\}$ and $I_2=\{(m,p),(p,p),(p,q)\}$ are not compatible, which is witnessed by the pair of $(n,n) \in I_1$ and $(m,p) \in I_2$.
If one of them was absent, they were compatible.}
\end{figure*}
Suppose that $s$ has two children $t_1$ and $t_2$, where $X_s = X_{t_1} = X_{t_2}$.
We say that $A_1=(\pi_1,I_1,c_1) \in \mscr{I}_{t_1}$ and $A_2=(\pi_2,I_2,c_2) \in \mscr{I}_{t_2}$ are \emph{compatible} if $\pi_1 = \pi_2$ and there are no $(p_1,q_1)\in I_1$ and $(p_2,q_2) \in I_2$ such that $\its{\pi_1}{(p_1,q_1)}{(p_2,q_2)}$.
Figure~\ref{fig:interval_join} illustrates an incompatible pair.
If $A_1$ and $A_2$ are not compatible, any interval representation $\rho$ on $X_{\le s}$ which extends $\rho_1$ and $\rho_2$ will make two vertices adjacent which are not adjacent in the input graph $G$ for any interval representations $\rho_i$ on $X_{\le t_i}$ of which $A_i$ is the abstraction for $i=1,2$.
For each compatible pair $(A_1,A_2)$, one can find an interval representation $\rho$ on $X_{\le s}$ that forms a subgraph of $(X_{\le s},E_{\le s})$ which extends some $\rho_1$ and $\rho_2$ whose abstractions are $A_1$ and $A_2$, respectively.
We add the triple $(\pi_1,\,I_1 \cup I_2,\,c_1+c_2)$ to $\mscr{I}'_{s}$.
We obtain $\mscr{I}_{s}$ by reducing $\mscr{I}'_{s}$.

\begin{theorem}\label{thm:interval}
The edge deletion problem for interval graphs can be solved in $O(|V| (2k)! \cdot 2^{7.76 k} \mrm{poly}(k))$ time 
%where $N=(2k)! \cdot 2^{3.38k} = (2k)! \cdot 2^{O(k)}$ 
for the treewidth $k$ of $G$ and some polynomial function $\mrm{poly}$.
If $k$ is the pathwidth, it can be solved in $O(|V| (2k)! \cdot 2^{3.38 k} \mrm{poly}(k))$ time.
\end{theorem}
\begin{proof}
Let $k$ be the maximum size of the assigned set $X_s$ to a node of a nice tree-decomposition.
We first estimate the number $N$ of possible forbidden interval sets $I$ for a fixed $\pi$ in $(\pi,I,c) \in \mscr{I}_s$.
Recall that $I$ must satisfy that
\begin{itemize}
\item $(p,q) \in I$ implies $p \le_{\pi} q$,
\item $(p,q_1),(p,q_2) \in I$ and $p \neq q_1,q_2$ implies $q_1=q_2$, 
\item $(p_1,q),(p_2,q) \in I$ and $q \neq p_1,p_2$ implies $p_1=p_2$.
\end{itemize}
That is, there are at most three intervals in $I$ that involves each $p \in \mathit{LR}_{X_s}-\{\top\}$:
\begin{enumerate}
	\item ending a forbidden interval started earlier: $(q,p) \in I$ for some $q <_{\pi} p$,
	\item starting and ending a minimal forbidden interval: $(p,p) \in I$,
	\item starting a new forbidden interval: $(p,q) \in I$ for some $q >_{\pi} p$.
\end{enumerate}
Those possibilities are neither mutually exclusive nor independent.
To count the number of possible forbidden intervals inductively, hereafter we will also count interval sets that may contain a right-unbounded interval.
Let $A_\mrm{b}(n)$ and $A_\mrm{u}(n)$ be the numbers of possible ways of drawing forbidden intervals with $n$ points in addition to $\bot$ where the last intervals are right-bounded and right-unbounded, respectively.
%The first case is possible only when we have an open interval before $p$
%and the second and third cases are possible only when we have no open interval before or the open interval is closed by the first case.
Solving the recurrence equations
\begin{align*}
	A_{\mrm{b}}(n+1) &= 2 A_\mrm{b}(n) + 2 A_\mrm{u}(n)\,, & A_{\mrm{b}}(0)=2\,, 
\\	A_{\mrm{u}}(n+1) &= 2 A_\mrm{b}(n) + 3 A_\mrm{u}(n)\,, & A_{\mrm{u}}(0)=1\,,
\end{align*}
we obtain $A_{\mrm{b}}(n) \in O(2^{2.19 n})$.
% \frac{\sqrt{33}+11}{4} \le 4.186141
Therefore, since we have $2k$ points, there are at most $N \in O(2^{4.38 k})$ possible sets of forbidden intervals.
There can be at most $(2k)!/2^k$ varieties of $\pi$.
Recall that $\mscr{I}_s$ is reduced, where if $\mscr{I}_s$ has two elements of the form $ (\pi,I,c)$ and $(\pi,I,c')$, then $c=c'$.
We conclude that each $\mscr{I}_s$ may contain at most $(2k)!/2^k N \in O(2^{3.38 k}(2k)!)$ elements.
% log_2(5) < 2.32193

Suppose $s$ is a forget node, whose child is $t$.
Then, from each element $(\pi,I,c) \in \mscr{I}_t$, we calculate just one element $(\pi',I',c') \in \mscr{I}'_s$.
Therefore, the calculation can be done in $O(N (2k)!/2^k \mrm{Poly}(k))$ time.

Suppose $s$ is an introduce node, whose child is $t$.
Then, each element $(\pi',I',c') \in \mscr{I}_s$ is derived from a unique element $(\pi,I,c) \in \mscr{I}_t$.
Therefore, the calculation can be done in $O(N (2k)!/2^k \mrm{Poly}(k))$ time.

Suppose $s$ is a join node with children $t_1$ and $t_2$.
Recall that $(\pi_1,I_1,c_1) \in \mscr{I}_{t_1}$ and $(\pi_2,I_2,c_2) \in \mscr{I}_{t_2}$ are compatible only when $\pi_1=\pi_2$.
Checking the compatibility of $(\pi,I_1,c_1)$ and $(\pi,I_2,c_2)$ and computing their ``join'' $(\pi,I_1 \cup I_2, c_1+c_2)$ takes $\mrm{Poly}(k)$ time.
Since we have at most $N^2$ pairs to examine for each $\pi$, it takes $O((2k)!/2^k N^2 \mrm{Poly}(k))$ time.

Since the nice tree-decomposition has $O(|V|)$ nodes, we obtain the conclusion.
%All in all, our algorithm runs in $O(|V| N^2 \mrm{poly}(k))$ time,
%If the tree-decoposition has no join nodes, then it will be $O(|V|N \mrm{poly}(k'))$ time.
\end{proof}
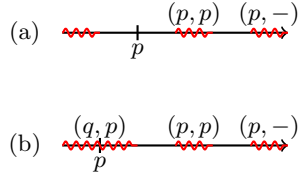
\begin{figure*}[t]\centering
	\begin{tikzpicture}[scale=0.5]\small
		\draw(-1,0) node{(a)};
		\draw[->,thick] (0,0)--(6,0);
		\point[thick]{2}{0}
		\draw (2,-0.5) node{$p$};
		\interval[finterval]{0}{1}{0}
		\interval[finterval]{3}{4}{0}
		\draw (3.5,0.5) node{$(p,p)$};
		\interval[finterval]{5}{6}{0}
		\draw (5.5,0.5) node{$(p,-)$};
\begin{scope}[shift={(0,-3)}]
		\draw(-1,0) node{(b)};
		\draw[->,thick] (0,0)--(6,0);
		\point[thick]{1}{0}
		\interval[finterval]{0}{2}{0}
		\draw (1,-0.5) node{$p$};
		\draw (1,0.5) node{$(q,p)$};
		\interval[finterval]{3}{4}{0}
		\draw (3.5,0.5) node{$(p,p)$};
		\interval[finterval]{5}{6}{0}
		\draw (5.5,0.5) node{$(p,-)$};
\end{scope}
	\end{tikzpicture}
\caption{\label{fig:interval_complexity}%
%Possible relations between a minimal interval $(p,\suc{\pi}(p))$ and $I$.
Possible ways to extend an interval set over $n$ points to one over $n+1$ points.
(a) If the last interval is right-bounded, there are four ways to add $p$ to the set.
(b) If the last interval is right-unbounded, there are five ways to add $p$ to the set.
}
\end{figure*}

%% file: permutation.tex
% !TEX root = main_arxiv.tex
\section{Finding a Largest Permutation Subgraph}
 A \emph{permutation representation} over a set $X$ is a pair $\pi=(\pi_1,\pi_2)$ of linear orders on $X^+ = X \cup \{\top,\bot\}$ such that $\bot <_{\pi_i} x <_{\pi_i} \top$ for all $x \in X$ and $i \in \{1,2\}$.
For two points $u$ and $v$ in $X^+$, we write $\prl{\pi}{u}{v}$ if $\pi_1$ and $\pi_2$ agree on the order of $u$ and $v$: that is,
either $u \le_{\pi_1} v$ and $u \le_{\pi_2} v$ or $v <_{\pi_1} u$ and $v <_{\pi_2} u$.
Otherwise, we write $\its{\pi}{u}{v}$ and say that $u$ and $v$ \emph{intersect} in $\pi$.
% We say that $(u,v)\in X \times X$ \emph{intersects} in $\pi$ if either $u <_{\pi_1} v$ and $v <_{\pi_2} u$ or $v <_{\pi_1} u$ and $u <_{\pi_2} v$.
 The \emph{permutation graph $G_\pi$ of a permutation representation $\pi$ on $V$} is $(V,\mcal{E}_\pi)$ where 
 \[
 \mcal{E}_\pi = \{\,\{u,v\} \subseteq V \mid \text{$u$ and $v$ intersect in $\pi$} \,\}
 \,.\]
 Figure~\ref{fig:perm} (a) shows an example of a permutation representation $\rho$ and (b) shows the induced permutation graph $G_\rho$.
 For each vertex $u$ of $V$, we put a point on each of the two parallel lines respecting the linear orders $\pi_1$ and $\pi_2$ and draw a line, called the $u$-line, between the points.
 Then we make an edge between $u$ and $v$ if and only if the $u$-line and the $v$-line intersect.
 This section gives an FPT algorithm for the edge deletion problem for permutation graphs.
 \begin{figure}
 \centering
% \begin{mysubfig}{\textwidth}
 	\centering
 	\begin{tikzpicture}[scale=0.5]\small
 		\draw(4, -1) node{(a)};
 		\draw[->](0.5,3)--(7.5,3);
 		\draw[->](0.5,0.5)--(7.5,0.5);
 		\draw(1, 3.5) node{$u_2$};
 		\draw(2, 3.5) node{$w_2$};
 		\draw(3, 3.5) node{$w_1$};
 		\draw(4, 3.5) node{$u_1$};
 		\draw(5, 3.5) node{$w_3$};
 		\draw(6, 3.5) node{$w_4$};
 		\draw(7, 3.5) node{$u_3$};
 		\draw(1, 0) node{$u_1$};
 		\draw(2, 0) node{$w_1$};
 		\draw(3, 0) node{$w_3$};
 		\draw(4, 0) node{$w_2$};
 		\draw(5, 0) node{$w_4$};
 		\draw(6, 0) node{$u_3$};
 		\draw(7, 0) node{$u_2$};
 		\draw[-](1, 3)--(7,0.5);
 		\draw[-](2, 3)--(4,0.5);
 		\draw[-](3, 3)--(2,0.5);
 		\draw[-](4, 3)--(1,0.5);
 		\draw[-](5, 3)--(3,0.5);
 		\draw[-](6, 3)--(5,0.5);
 		\draw[-](7, 3)--(6,0.5);
 		\draw(0, 3) node{$\rho_1$};
 		\draw(0, 0.5) node{$\rho_2$};
 	\end{tikzpicture}
 	\quad
 	\begin{tikzpicture}[scale=0.5, every node/.style=circle,fill=white,inner sep=0pt,minimum size=4mm]\small
 	\draw(2.25, -1) node{(b)};
 	\node[draw,fill] (u2) at (0,3) {$u_1$};
 	\node[draw,fill] (u1) at (2.25,3) {$u_2$};
 	\node[draw,fill] (u3) at (4.5,3) {$u_3$};
 	\node[draw,fill] (w1) at (0,0.5) {$w_1$};
 	\node[draw,fill] (w2) at (1.5,0.5) {$w_2$};
 	\node[draw,fill] (w3) at (3,0.5) {$w_3$};
 	\node[draw,fill] (w4) at (4.5,0.5) {$w_4$};
 	\draw[thick,-] (w1) to (w2);
 	\draw[thick,-] (w1) to (u1);
 	\draw[thick,-] (w1) to (u2);
 	\draw[thick,-] (w2) to (u1);
 	\draw[thick,-] (w2) to (u2);
 	\draw[thick,-] (w2) to (w3);
 	\draw[thick,-] (u1) to (u2);
 	\draw[thick,-] (u1) to (u3);
 	\draw[thick,-] (u1) to (w3);
 	\draw[thick,-] (u1) to (w4);
 	\end{tikzpicture}
 %	\caption{}
 %\end{mysubfig}
 %\quad
 %\begin{mysubfig}{0.4\textwidth}
 %	\centering
 	\quad
 	\begin{tikzpicture}[scale=0.5]\small
 		\draw(4, -1) node{(c)};
		\fill[red!10] (2,3) -- (3.5,1.125) -- (5,3);
		\fill[red!10] (2,0.5) -- (2.667,2.167) -- (4,0.5);
		\fill[red!10] (5.9,3) -- (4.9,0.5) -- (5.1,0.5) -- (6.1,3);
 		\draw[->](0.5,3)--(7.5,3);
 		\draw[->](0.5,0.5)--(7.5,0.5);
 		\draw(1, 3.5) node{$u_2$};
 		\draw(2, 3.5) node[forgotten]{$w_2$};
 		\draw(3, 3.5) node[forgotten]{$w_1$};
 		\draw(4, 3.5) node{$u_1$};
 		\draw(5, 3.5) node[forgotten]{$w_3$};
 		\draw(6, 3.5) node[forgotten]{$w_4$};
 		\draw(7, 3.5) node{$u_3$};
 		\draw(1, 0) node{$u_1$};
 		\draw(2, 0) node[forgotten]{$w_1$};
 		\draw(3, 0) node[forgotten]{$w_3$};
 		\draw(4, 0) node[forgotten]{$w_2$};
 		\draw(5, 0) node[forgotten]{$w_4$};
 		\draw(6, 0) node{$u_3$};
 		\draw(7, 0) node{$u_2$};
 		\draw[-](1, 3)--(7,0.5);
 		\draw[forgotten](2, 3)--(4,0.5);
 		\draw[forgotten](3, 3)--(2,0.5);
 		\draw[-](4, 3)--(1,0.5);
 		\draw[forgotten](5, 3)--(3,0.5);
 		\draw[forgotten](6, 3)--(5,0.5);
 		\draw[-](7, 3)--(6,0.5);
 		\draw(0, 3) node{$\rho_1$};
 		\draw(0, 0.5) node{$\rho_2$};
 	\end{tikzpicture}
% 	\begin{tikzpicture}[scale=0.5]\small
% 		\draw(3, -1) node{\bf\textsf{(c)}};
% %		
% 		\draw(0, 3) node{$\pi_1$};
% 		\draw(0, 0.5) node{$\pi_2$};
% 		\draw[->](0.5,3)--(5.5,3);
% 		\draw[->](0.5,0.5)--(5.5,0.5);
% 		\draw(1, 3.5) node{$u_2$};
% 		\draw(2, 3.5) node{};
% 		\draw(3, 3.5) node{$u_1$};
% 		\draw(4, 3.5) node{};
% 		\draw(5, 3.5) node{$u_3$};
% 		\draw(1, 0) node{$u_1$};
% 		\draw(2, 3.5) node[forgotten]{\tiny$\begin{array}{l}w_1,\\w_2\end{array}$};
% 		\draw(4, 3.5) node[forgotten]{\tiny$\begin{array}{l}w_3,\\w_4\end{array}$};
% 		\draw(2.5, 0) node[forgotten]{\tiny$\begin{array}{l}w_1,w_2,\\w_3,w_4\end{array}$};
% 		\draw(4, 0) node{$u_3$};
% 		\draw(5, 0) node{$u_2$};
% 		\draw[-](1, 3)--(5,0.5);
% 		\draw[forgotten](2, 3)--(2.5,0.5);
% 		\draw[-](3, 3)--(1,0.5);
% 		\draw[forgotten](4, 3)--(2.5,0.5);
% %		\draw[forgotten](4, 3)--(4,0.5);
% 		\draw[-](5, 3)--(4,0.5);
% 	\end{tikzpicture}
 %	\caption{}
% \end{mysubfig}
 \caption{\label{fig:perm} (a) Permutation representation $\rho$. (b) Permutation graph $G_\rho$. (c) Illustration of the abstraction $\mscr{A}(\rho,s)=(\pi,I,c)$  where $X_s=\{u_1,u_2,u_3\}$ and $X_{\le s}-X_s = \{w_1,w_2,w_3,w_4\}$. The forbidden areas are represented by $I= \{(u_2,u_1,u_1,u_1),(u_1,u_1,u_1,u_1)\}$, where the intersection closure $[w_1]_\rho^s=\{w_1,w_2,w_3\}$ is anchored to $(u_2,u_1,u_1,u_1)$ and $[w_4]_\rho^s=\{w_4\}$ is anchored to $(u_1,u_1,u_1,u_1)$ in $\mscr{A}(\rho,s)$.}
 \end{figure}
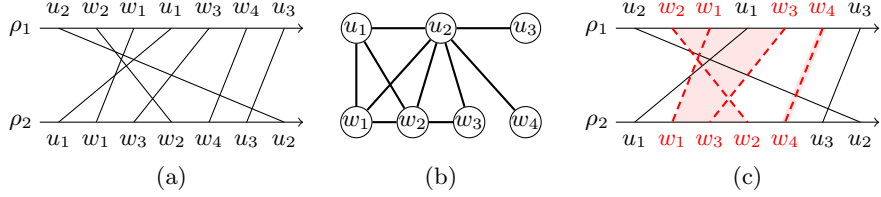

\subsection{Algorithm Invariant}
We anchor the areas bordered with the lines of forgotten vertices $X_{\le s} - X_s$ to the points of the current bag $X_s$.
% For a permutation representation $\rho = (\rho_1,\rho_2)$ over $X_{\le s}$, we anchor forgotten areas to the current points.
For each $w \in X_{\le s} - X_s$, let the \emph{intersection closure of $w$ (w.r.t.\ $\rho$ and $s$)} be the smallest set $[w]_\rho^s$ such that 
\begin{itemize}
\item $w \in [w]_\rho^s$ and
\item if $w_1 \in [w]_\rho^s$, $w_2 \in X_{\le s} - X_s$, and $\its{\rho}{w_1}{w_2}$, then $w_2\in [w]_\rho^s$.
\end{itemize}
%We usually omit the superscript $s$ and the subscript $\rho$ if they are understood from the context.
Each closure forms a forbidden area in the sense that we may introduce no new lines that intersect the area.

For each permutation representation $\rho = (\rho_1,\rho_2)$ over $X_{\le s}$, we define the \emph{abstraction $\mscr{A}(\rho,s) = ((\pi_1,\pi_2),I,c)$ of $\rho$ for $s$} as follows: 
\begin{itemize}
	\item $\pi_1$ and $\pi_2$ are the restrictions of $\rho_1$ and $\rho_2$ to $X_s$, respectively,
	\item $I = \{\, (p_1,q_1,p_2,q_2) \in (X_s^+)^4 \mid $ there is $w \in X_{\le s} - X_s$ s.t.\ for $i \in \{1,2\}$,\\ \hfill$p_i <_{\rho_i} \min_{\rho_i}[w]_{\rho}^s <_{\rho_i} \suc{\pi_i}(p_i)$ and $q_i <_{\rho_i} \max_{\rho_i}[w]_{\rho}^s <_{\rho_i} \suc{\pi_i}(q_i)\,\}$,
	\item $c = |E_{\le s} - \mcal{E}_\rho - E_s| %$\\\phantom{$c$}${}
	=|\{\, \{u,v\} \in E_{\le s} \mid \text{$\{u,v\} \nsubseteq X_s$ and $\its{\rho}{u}{v}$} \,\}|$.
\end{itemize}
Intuitively, $I$ anchors occurrences of forgotten vertices in $\rho$ to occurrences of active vertices in $\pi=(\pi_1,\pi_2)$.
Figure~\ref{fig:perm} (c) illustrates an example of forbidden areas.
We note that, two distinct elements $(p_1,q_1,p_2,q_2) \in I$ and $(p_1',q_1',p_2',q_2') \in I$ cannot ``intersect'', due to the definition of intersection closures.
In other words, there are no $(p_1,q_1,p_2,q_2) , (p_1',q_1',p_2',q_2') \in I$ such that $p_i <_{\pi_i} q_i'$ and $p_{j}' <_{\pi_j} q_j$ for some $i,j \in \{1,2\}$.
Elements of $I$ will be linearly ordered by extending $\pi$ so that $(p_1,q_1,p_2,q_2) <_\pi (p_1',q_1',p_2',q_2') $ if and only if $p_1 < q_1'$ or $p_2 < q_2'$.
%Suppose we extend $\pi$ to $\pi'=(\pi_1',\pi_2')$ introducing a new vertex $x$.
%%Let $x_1 = \pred{\pi'_1}(x)$ and  $x_2 = \pred{\pi'_2}(x)$.
%Then $x$ will intersect some forgotten vertex if there is $(p_1,q_1,p_2,q_2) \in I$ such that
%\begin{itemize}
%\item $\suc{\pi_i'}(p_i) <_{\pi_i'} x$ and $x <_{\pi_j'} q_j$ for some $i,j \in \{1,2\}$.
%%\item $p_i <_{\pi_1} x_i <_{\pi_1} q_i$ for some $i \in \{1,2\}$, or
%%\item $x <_{\pi_1} p_1$ and $\suc{\pi_2}(p_2) <_{\pi_2} x$ for some $(p_1,p_2) \in Z$, or
%\end{itemize}
The integer $c$ is the number of the edges deleted from $(X_{\le s}, E_{\le s})$ except those among active vertices.

%The abstraction $\mscr{A}(\rho,X)$ preserves enough information about the alignment of forgotten vertices in $W - X$ for our purpose.
We say that $(\pi,I,c)$ \emph{dominates} $(\pi',I',c')$ if and only if $\pi=\pi'$, $I \sqsubseteq_\pi I'$ and $c \le c'$,
where we write $I \sqsubseteq_\pi I'$ if every forbidden area of $I$ is inside some forbidden area of $I'$:
i.e., for every $(p_1,q_1,p_2,q_2) \in I$, there is $(p_1',q_1',p_2',q_2') \in I'$ such that
 $p_1' \le_{\pi_1} p_1$, $q_1 \le_{\pi_1} q_1'$, $p_2' \le_{\pi_2} p_2$, and $q_2 \le_{\pi_2} q_2'$.
\begin{lemma}
Suppose $\mscr{A}(\rho,s)=(\pi,I,c)$ dominates $\mscr{A}(\rho',s)=(\pi',I',c')$.
For any extension $\sigma'$ of $\rho'$ such that $\mcal{E}_{\sigma'} \subseteq E_{\le s}$, there is an extension $\sigma$ of $\rho$ such that $\mcal{E}_{\sigma'} \subseteq \mcal{E}_{\sigma} \subseteq E_{\le s}$.
\end{lemma}
We call a set of abstractions \emph{reduced} if it has no pair of distinct elements of which one dominates the other.
The reduced form of a set of abstractions is obtained by removing abstractions that are dominated by others.
%One can effectively ``reduce'' a set of abstractions.
Our algorithm for \textsc{Permutation-Edge-Deletion} calculates, for each node $s$ of the tree-decomposition, a reduced set $\mscr{I}_s$ of abstractions of permutation representations of permutation subgraphs of $(X_{\le s},E_{\le s})$ for $X_s$ satisfying the following invariant properties.
%More precisely, our algorithm works with the following invariant properties:
\begin{condition}\label{cond:invariant_p}{\ }
\begin{itemize}
 \item Every element of $\mscr{I}_s$ is the abstraction of some permutation representation of a permutation subgraph of $(X_{\le s}, E_{\le s})$ for $X_s$,
 \item Any permutation representation $\rho$ of any permutation subgraph of $(X_{\le s}, E_{\le s})$ has an element of $\mscr{I}_s$ that dominates its abstraction $\mscr{A}(\rho,s)$.
 \end{itemize}
\end{condition}
 Clearly $\mscr{I}_s$ satisfies the above condition if and only if so is its reduced form.
 We make each $\mscr{I}_s$ reduced.
 Particularly if $s$ is the root node, we have $\mscr{I}_s=\{((o,o), \{(\bot,\bot,\bot,\bot)\},c)\}$, where $o$ is the trivial order over $\{\bot,\top\}$ such that $\bot <_o \top$ and $c$ is the least number such that one can obtain a permutation subgraph by removing $c$ edges from $G$.
 That is, the number $c$ will be the solution to our problem.
 If $s$ is a leaf, by definition $\mscr{I}_s = \{((o,o),\emptyset,0)\}$.
 Our algorithm computes those values $\mscr{I}_s$ recursively from leaves to the root.
 In what follows, we show how to calculate $\mscr{I}_{s}$ from $\mscr{I}_t$ for child(ren) $t$ of $s$, while preserving the invariant (Condition~\ref{cond:invariant_p}).

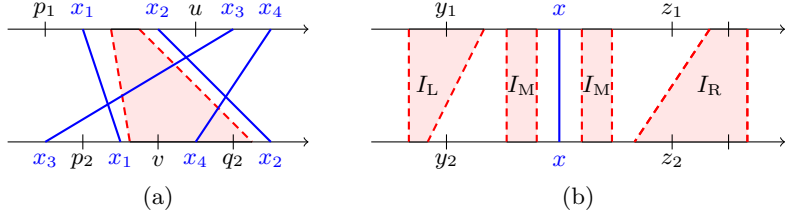
\begin{figure}\centering
\newcommand{\forbiddenarea}[4]{%
	\fill[red!10] (#1,3) -- (#3,3) -- (#4,0) -- (#2,0);
	\draw[forgotten](#1, 3)--(#2,0);
	\draw[forgotten](#3, 3)--(#4,0);
	\draw[-](#1,3)--(#3,3);
	\draw[-](#2,0)--(#4,0);
}
%\begin{mysubfig}{0.36\textwidth}{Concerning $(p_1,p_1,p_2,q_2) \in I$, one may introduce a new line like $x_1$ or $x_2$ but not $x_3$ or $x_4$.}
	\begin{tikzpicture}[scale=0.5]\small
 		\draw(4, -1.5) node{(a)};
		\forbiddenarea{2.75}{3.25}{3.5}{6.5}
		\draw(-0.5, 0) node{};%dummy
		\draw[->](0,0)--(8.0,0);
		\draw[->](0,3)--(8.0,3);
		\point{1}{3}
		\point{5.0}{3}
		\draw(1, 3.5) node{$p_1$};
		\draw(5.0, 3.5) node{$u$};
		\point{2.0}{0}
		\point{4.0}{0}
		\point{6.0}{0}
		\draw(2.0, -0.5) node{$p_2$};
		\draw(4.0, -0.5) node{$v$};
		\draw(6.0, -0.5) node{$q_2$};
		\draw[focused](2.0, 3)--(3.0,0);
		\draw(2.0, 3.5) node[focused]{$x_1$};
		\draw(3.0, -0.5) node[focused]{$x_1$};
		\draw[focused](4.0, 3)--(7,0);
		\draw(4.0, 3.5) node[focused]{$x_2$};
		\draw(7.0, -0.5) node[focused]{$x_2$};
		\draw[focused](6.0, 3)--(1,0);
		\draw(6.0, 3.5) node[focused]{$x_3$};
		\draw(1, -0.5) node[focused]{$x_3$};
		\draw[focused](7.0, 3)--(5,0);
		\draw(7.0, 3.5) node[focused]{$x_4$};
		\draw(5, -0.5) node[focused]{$x_4$};
	\end{tikzpicture}
%\end{mysubfig}
\quad\quad
%\begin{mysubfig}{0.57\textwidth}{Relation between newly introduced vertex $x$ and forbidden areas. Since forbidden areas should not intersect with $x$, they are grouped into three: ones ($I_\mrm{L}$) that must be left to $x$, ones ($I_\mrm{R}$) that must be right to $x$, and the other ($I_\mrm{M}\subseteq\{(y_1,y_1,y_1,y_1)\}$).}
	\begin{tikzpicture}[scale=0.5]\small
 		\draw(5.5, -1.5) node{(b)};
		\draw[->](0,0)--(11,0);
		\draw[->](0,3)--(11,3);
		\forbiddenarea{1}{1}{3}{1.5}
		\forbiddenarea{3.6}{3.6}{4.4}{4.4}
		\forbiddenarea{5.6}{5.6}{6.4}{6.4}
		\forbiddenarea{9}{7}{10}{10}
		\draw(1.5, 1.5) node{$I_\mrm{L}$};
		\draw(4, 1.5) node{$I_\mrm{M}$};
		\draw(6, 1.5) node{$I_\mrm{M}$};
		\draw(9, 1.5) node{$I_\mrm{R}$};
		\point{2}{3}
		\point{8}{3}
		\point{9.5}{3}
		\point{2}{0}
		\point{8}{0}
		\point{9.5}{0}
		\draw(2, 3.5) node{$y_1$};
		\draw(8, 3.5) node{$z_1$};
		\draw(2, -0.5) node{$y_2$};
		\draw(8, -0.5) node{$z_2$};
		\draw[focused](5, 3)--(5, 0);
		\draw(5, 3.5) node[focused]{$x$};
		\draw(5, -0.5) node[focused]{$x$};
		\draw[forgotten](9, 3)--(7,0);
		\draw[forgotten](10, 3)--(10,-0);
	\end{tikzpicture}
%\end{mysubfig}
\caption{\label{fig:permalgo1}Introduce Node. (a) Extending $\pi$ by introducing $x_1$ or $x_2$ respects $I$ and introducing $x_3$ or $x_4$ does not. We see $\suc{\pi'_1}(p_1) <_{\pi'_1} x_3 <_{\pi'_2} q_2 $ and $\suc{\pi'_2}(p_2) <_{\pi'_2} x_4 <_{\pi'_2} q_2 $.
(b) Relation between newly introduced vertex $x$ and forbidden areas. Since forbidden areas should not intersect with $x$, they are grouped into three: ones ($I_\mrm{L}$) that must be left to $x$, ones ($I_\mrm{R}$) that must be right to $x$, and the other ($I_\mrm{M}\subseteq\{(y_1,y_1,y_1,y_1)\}$).}
\end{figure}

\subsection{Algorithm}
\paragraph*{Leaf Node:} If $s$ is a leaf, we let $\mscr{I}_s = \{((o,o),\emptyset,0)\}$.

%\noindent \textbf{\textsf{Introduce Node:}}
\paragraph*{Introduce Node:}
Suppose $s$ is an introduce node. It has just one child $t$ such that $X_{s} = X_t \cup \{x\}$.
Figure~\ref{fig:permalgo1} would help to understand the behavior of our algorithm in this case.
For each $(\pi,I,c)$ in $\mscr{I}_t$, we say that an extension $\pi'$ of $\pi$ to $X_s$ \emph{respects} $E$ and $I$ precisely in the cases where
\begin{itemize}
		\item if $\{x,u\} \notin E$ for $u \in X_{t}$, then $\prl{\pi'}{u}{x}$,
%		\item if $(p_1,q_1,p_2,q_2) \in I$, then $x_i <_{\pi_i} p_i$ or $q_{j} <_{\pi_j} x_j$ for any combination of $i,j \in \{1,2\}$ for $x_1 = \pred{\pi'_1}(x)$ and $x_2 = \pred{\pi'_2}(x)$,
		\item there are no $(p_1,q_1,p_2,q_2) \in I$ and $i,j \in \{1,2\}$ such that
 			$\suc{\pi'_i}(p_i) <_{\pi'_i} x$ and $x <_{\pi'_j} q_j$, % for $x_1 = \pred{\pi'_1}(x)$ and $x_2 = \pred{\pi'_2}(x)$,
	\end{itemize}
respectively.
If $\pi'$ does not respect $E$, we would wrongly make two non-adjacent vertices $u \in X_s$ and $x$ of $G$ intersect in $\pi'$, which gives a non-subgraph of the input graph.
Otherwise, they will be non-adjacent.
Since $\{x,w\} \notin E$ for any forgotten vertex $w \in X_{\le s} - X_s$, we must avoid $x$ and $w$ intersect in an extension of $\rho$.
This requires $\pi'$ to respect $I$.
If $\pi'$ respects $I$, there are $\rho$ such that $\mscr{A}(\rho,t)=(I,\pi,c)$ and an extension $\rho'$ of both $\rho$ and $\pi'$ such that $\mcal{E}_{\rho'} \subseteq E_{\le s}$.
We note that here we use $\suc{\pi'_i}(p_i) $ rather than $p_i$ to judge whether $\pi'$ respects $I$.
This is because the quadruple $(p_1,q_1,p_2,q_2) \in I$ consists of anchors rather than the exact points of the forbidden area.

For each extension $\pi'$ respecting $E$ and $I$, we add at most two triples to $\mscr{I}'_s$ as described below.
Since $X_s \supseteq X_t$, $c$ need not be updated by definition.

Some forbidden areas that had an anchor $y_i=\pred{\pi_i'}(x)$ may be re-anchored to $x$.
Forbidden areas are partitioned into three, where the ones in $I_\mrm{L}$ and $I_\mrm{R}$ are certainly left to and right to $x$, respectively, while $(y_1,y_1,y_2,y_2)$ can go to the left or right to $x$: 
\begin{align*}
	I_\mrm{L} &= \{\, (p_1,q_1,p_2,q_2) \in I \mid p_1 <_{\pi_1} y_1 \text{ or } p_2 <_{\pi_2} y_2 \,\} \,,
\\	I_\mrm{R} &= \{\, (p_1,q_1,p_2,q_2) \in I \mid y_1 <_{\pi_1} q_1 \text{ or } y_2 <_{\pi_2} q_2\,\} \,,
\\	I_\mrm{M} &= \{\, (p_1,q_1,p_2,q_2) \in I \mid p_1=y_1=q_1 \text{ and } p_2=y_2=q_2 \,\} \,.
\end{align*}
Note that $I_\mrm{M}$ is either empty or a singleton and that $I_\mrm{L}$ and $I_\mrm{R}$ are mutually exclusive, since $\pi'$ respects $I$.
If a forbidden area has a point anchored to $y_i$ and is right to $x$, it will be re-anchored:
\begin{align*}
	I_\mrm{R}' &= \{\, (p_1[x/y_1],q_1[x/y_1],p_2[x/y_2],q_2[x/y_2]) \mid (p_1,q_1,p_2,q_2) \in I_\mrm{R} \,\} \,,
\\	I_\mrm{M}' &= \{\, (p_1[x/y_1],q_1[x/y_1],p_2[x/y_2],q_2[x/y_2]) \mid (p_1,q_1,p_2,q_2) \in I_\mrm{M} \,\} \,.
\end{align*}
Corresponding to the possible choices, we put both $(\pi',I_\mrm{L} \cup I_\mrm{M} \cup I_\mrm{R}',c)$ and $(\pi',I_\mrm{L} \cup I_\mrm{M}' \cup I_\mrm{R}',c)$ into $\mscr{I}'_s$.
We will not add $(\pi',I_\mrm{L} \cup I_\mrm{M} \cup I_\mrm{M}' \cup I_\mrm{R}',c)$ since it is dominated by the other two.

We then obtain $\mscr{I}_{s}$ by reducing $\mscr{I}'_{s}$.

%\medskip
%\noindent \textbf{\textsf{Forget Node:}}
\paragraph*{Forget Node:}
Suppose $s$ is a forget node. It has just one child $t$ such that $X_{t} = X_s \cup \{x\}$.
For each $(\pi,I,c)$ in $\mscr{I}_t$, we add the following triple $(\pi',I',c')$ to $\mscr{I}'_{s}$ where
\begin{itemize}
	\item $\pi'$ is the restriction of $\pi$ for $X_s$,
	\item $c' = c + |\{\, u \in X_s \mid \text{$ \{x,u\} \in E_{\le s}$ and $\prl{\pi}{u}{x}$} \,\}|$.
\end{itemize}
Let
\[
	I_\mrm{X} = \{\, (p_1,q_1,p_2,q_2) \in I \mid p_i <_{\pi_i} x \text{ and } x \le_{\pi_j} q_j \text{ with } \{i,j\} = \{1,2\}\,\}
%	I_\mrm{X} = \{\, (p_1,q_1,p_2,q_2) \in I \mid \suc{\pi_i}(p_i) <_{\pi_i} x \text{ and }x <_{\pi_j} q_j \text{ for some } i,j \in \{1,2\}\,\}
\]
be the set of forbidden areas of $I$ intersecting with $x$. % re-anchored to points of $X_s$.
We note that the asymmetry between $p_i <_{\pi_i} x$ and $x \le_{\pi_j} q_j $ is due to the fact that the actual forbidden area starts after $p_1$ and $p_2$ and ends after $q_1$ and $q_2$, while the points of $x$ are exact.

The new forbidden area formed by $x$ and $I_\mrm{X}$ will be $\alpha = (y_1,z_1,y_2,z_2)$ where for $x_1 = \pred{\pi_1}(x)$ and $x_2 = \pred{\pi_2}(x)$,
\begin{align*}
	y_1 &= \min_{\pi_1} (\{x_1\} \cup \{\, p_1[x_1/x] \mid (p_1,q_1,p_2,q_2) \in I_\mrm{X} \,\})\,,
\\	z_1 &= \max_{\pi_1} (\{x_1\} \cup \{\, q_1[x_1/x] \mid (p_1,q_1,p_2,q_2) \in I_\mrm{X} \,\})\,,
\\	y_2 &= \min_{\pi_2} (\{x_2\} \cup \{\, p_2[x_2/x] \mid (p_1,q_1,p_2,q_2) \in I_\mrm{X} \,\})\,,
\\	z_2 &= \max_{\pi_2} (\{x_2\} \cup \{\, q_2[x_2/x] \mid (p_1,q_1,p_2,q_2) \in I_\mrm{X} \,\})\,.
\end{align*}
We then define $I'$ by 
\begin{align*}
	I' &= \{\, (p_1[x_1/x],q_1[x_1/x],p_2[x_2/x],q_2[x_2/x]) \mid (p_1,q_1,p_2,q_3) \in I-I_\mrm{X} \,\} \cup \{\, \alpha \,\} 
\end{align*}
This is illustrated in Figure~\ref{fig:perm_forget}.
\begin{figure}\centering
\newcommand{\forbiddenarea}[4]{%
	\fill[red!10] (#1,3) -- (#3,3) -- (#4,0) -- (#2,0);
	\draw[forgotten](#1, 3)--(#2,0);
	\draw[forgotten](#3, 3)--(#4,0);
	\draw[-](#1,3)--(#3,3);
	\draw[-](#2,0)--(#4,0);
}
	\begin{tikzpicture}[scale=0.5]\small
		\draw[->](0,0)--(10,0);
		\draw[->](0,3)--(10,3);
		\forbiddenarea{2.0}{0.5}{2.5}{1.5}
		\forbiddenarea{4.75}{2.5}{5.5}{3.5}
		\forbiddenarea{7}{5}{7.5}{7.5}
		\forbiddenarea{8.5}{8.5}{9}{9}
		\point{1}{3}
		\point{1}{0}
		\point{4}{3}
		\point{6}{3}
		\point{4}{0}
		\point{6}{0}
		\draw(4, 3.5) node{$m$};
		\draw(6, 3.5) node{$n$};
		\draw(1, -0.5) node{$p$};
		\draw(4, -0.5) node{$q$};
		\draw(1, 3.5) node{$x_1$};
		\draw(6, -0.5) node{$x_2$};
		\draw(3, 3.5) node[focused]{$x$};
		\draw(7, -0.5) node[focused]{$x$};
		\draw[focused](3, 3)--(7,0);
		\draw(4, 1.5) node{$I_\mrm{X}$};
		\draw(6.75, 1.5) node{$I_\mrm{X}$};
		\draw[->,ultra thick](10.5,1.5)--(11.5,1.5);
		\begin{scope}[shift={(12,0)}]
		\fill[red!10] (3,3) -- (7,3) -- (7,0);
		\fill[red!10] (3,0) -- (7,0) -- (5.5,3);
		\draw[->](0,0)--(10,0);
		\draw[->](0,3)--(10,3);
		\forbiddenarea{2.0}{0.5}{2.5}{1.5}
		\forbiddenarea{4.75}{2.5}{5.5}{3.5}
		\forbiddenarea{7}{5}{7.5}{7.5}
		\forbiddenarea{8.5}{8.5}{9}{9}
		\point{1}{3}
		\point{1}{0}
		\point{4}{3}
		\point{6}{3}
		\point{4}{0}
		\point{6}{0}
		\draw(4, 3.5) node{$m$};
		\draw(6, 3.5) node{$n$};
		\draw(1, -0.5) node{$p$};
		\draw(4, -0.5) node{$q$};
		\draw(1, 3.5) node{$x_1$};
		\draw(6, -0.5) node{$x_2$};
		\draw(3, 3.5) node[forgotten]{\scriptsize$(x)$};
		\draw(7, -0.5) node[forgotten]{\scriptsize$(x)$};
		\draw[forgotten](3, 3)--(7,0);
		\end{scope}
	\end{tikzpicture}
\caption{\label{fig:perm_forget} Forget Node. The forbidden areas in $I_\mrm{X}=\{(m,m,p,p),(n,n,q,x)\}$ will be merged into the single area $\alpha=(x_1,n,p,x_2)$ by forgetting $x$.}
\end{figure}
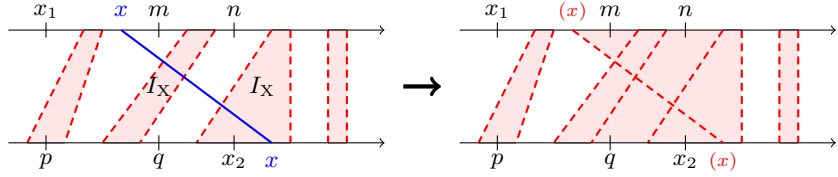
%Since we are forgetting $x$, forgotten vertices in $X_{\le t} - X_t$ that are anchored to $x$ in $I$ will be anchored to the predecessor of $x$ in $I'$.
%Moreover, the new forgetting vertex $x$ will be anchored to $(\pred{\pi_1}(x),\pred{\pi_2}(x))$ in $I'$.
%We update $c$ taking the newly forgotten $x$ into account.
We obtain $\mscr{I}_{s}$ by reducing $\mscr{I}'_{s}$.

%\medskip
%\noindent \textbf{\textsf{Join Node:}}
\paragraph*{Join Node:}
Suppose $s$ has two children $t$ and $t'$, where $X_{s} = X_{t} = X_{t'}$.
We say that a pair of $(\pi,I,c) \in \mscr{I}_{t}$ and $(\pi',I',c') \in \mscr{I}_{t'}$ is \emph{compatible} precisely in the case where
\begin{itemize}
\item $\pi = \pi'$, say $\pi=\pi'=(\pi_1,\pi_2)$, and
\item no members of $I$ and $I'$ intersect: that is,
there are no pairs of $(p_1,q_1,p_2,q_2) \in I$ and $(p_1',q_1',p_2',q_2') \in I'$ such that
$p_i <_{\pi_i} q_i'$ and $p_{j}' <_{\pi_j} q_j$ for some $i,j \in \{1,2\}$.
%		\item there are no $(p_1,q_1,p_2,q_2) \in I$, $i,j \in \{1,2\}$ such that
%			$p_i <_{\pi_i} x_i$ and $x_{j} <_{\pi_j} q_j$ for $x_1 = \pred{\pi'_1}(x)$ and $x_2 = \pred{\pi'_2}(x)$,
\end{itemize}
If they are compatible, we add the triple $((\pi_1,\pi_2),\,I \cup I',\,c + c')$ to $\mscr{I}'_{s}$.
We then obtain $\mscr{I}_{s}$ by reducing $\mscr{I}'_{s}$.

\begin{theorem}
The edge deletion problem for permutation graphs can be solved in $O(|V| (k!)^2 N^2 \mrm{poly}(k))$ time 
where $N=2^{7.34k}$ for the treewidth $k$ of $G$.
If $k$ is the pathwidth, it can be solved in $O(|V| (k!)^2 N \mrm{poly}(k))$ time.
\end{theorem}

\begin{proof}
Let $k$ be the maximum size of the assigned set $X_s$ to a node of the tree-decomposition.
There are at most $(k!)^2$ variants of $\pi$ of abstractions $(\pi,I,c)$.

In order to count the number $N$ of possible forbidden interval sets $I$ for a fixed $\pi$ in $(\pi,I,c) \in \mscr{I}_s$, we give a transformation of $I$ below.
Let $I_i$ for $i \in \{1,2\}$ be the interval projections of $I$:
\[
	I_i = \{\, (p_i,q_i) \mid (p_1,q_1,p_2,q_2) \in I\,\}
\,.\]
The proof of Theorem~\ref{thm:interval} has shown that there can be at most $2^{2.19 k}$ possibilities for each $I_i$.
However, since there is no one-one corresponding between $I_1$ and $I_2$, the pair $(I_1,I_2)$ is not informative enough to recover $I$.
For example, there can be distinct $(p,q)$ and $(p',q')$ in $I_1$ such that $(p,q,n,n),(p',q',n,n) \in I$.
We call an interval $(p,p) \in I_1$ a \emph{hub} if there are two (or more) distinct intervals $(p_2,q_2),(p_2',q_2') \in I_2$ such that $(p,p,p_2,q_2),(p,p,p_2',q_2') \in I$.
We call an interval $(p,q) \in I_1$ a \emph{terminal} if it is not a hub and there is no $n$ such that $(p,q,n,n),(\suc{\pi_1}(p,q),n,n) \in I$ where $\suc{\pi_1}(p,q) \in I_1$ is the successor of $(p,q)$ in $I_1$.
We use the same names for intervals in $I_2$.
Let $I_1'$ enrich $I_1$ so that each interval $(p,q) \in I_1$ has a 3-valued variable that indicate whether it is a hub, a terminal, or neither.
Figure~\ref{fig:perm_complexity} illustrates how $I$ will be transformed into $I_1'$ and $I_2'$.

We will show that $N$ is at most the number of possible enriched sets $I_1'$ and $I_2'$. 
Algorithm~\ref{alg:reconstruction} recovers elements of $I$ from left to right using $I_1'$ and $I_2'$, where elements of $I_1'$ and $I_2'$ are sorted and stored in stacks.
From the first intervals $(p_i,q_i) = \min_{\pi} I_i$ of the stacks, we obtain the first forbidden area $(p_1,q_1,p_2,q_2) \in I$.
If $(p_i,q_i)$ is a hub, we keep it as the \emph{active hub}.
Note that it is impossible that both $(p_1,q_1)$ and $(p_2,q_2)$ are hubs simultaneously.
Suppose we have reconstructed a forbidden area $(p_1,q_1,p_2,q_2) \in I$ with no active hub.
In this case, the next forbidden area will be $(\suc{\pi_1}(p_1,q_1),\suc{\pi_2}(p_2,q_2))$.
Suppose $(p_1,q_1)$ is the active hub.
If $(p_2,q_2)$ is neither a hub nor a terminal, then the next forbidden area will be $(p_1,q_1,\suc{\pi_2}(p_2,q_2))$ and we keep the active hub $(p_1,q_1)$.
If $(p_2,q_2)$ is a terminal, then the next forbidden area will be $(\suc{\pi_1}(p_1,q_1),\suc{\pi_2}(p_2,q_2))$ and the active hub will be null.
If $(p_2,q_2)$ is a hub, then the next forbidden area will be $(\suc{\pi_1}(p_1,q_1),p_2,q_2)$ and the new active hub is $(p_2,q_2) \in I_2$.
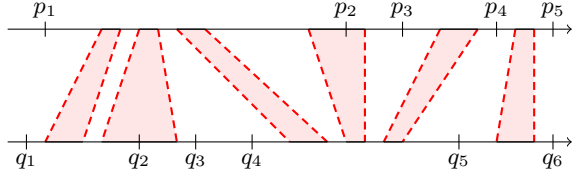
\begin{figure}\centering
\newcommand{\forbiddenarea}[4]{%
	\fill[red!10] (#1,3) -- (#3,3) -- (#4,0) -- (#2,0);
	\draw[forgotten](#1, 3)--(#2,0);
	\draw[forgotten](#3, 3)--(#4,0);
	\draw[-](#1,3)--(#3,3);
	\draw[-](#2,0)--(#4,0);
}
	\begin{tikzpicture}[scale=0.5]\small
		\forbiddenarea{3.5}{2}{4}{3}
		\forbiddenarea{4.5}{3.5}{5}{5.5}
		\forbiddenarea{5.5}{8.5}{6.25}{9.5}
		\forbiddenarea{9}{10}{10.5}{10.5}
		\forbiddenarea{12.5}{11}{13.5}{11.5}
		\forbiddenarea{14.5}{14}{15}{15}
%		\draw(-0.5, 0) node{};%dummy
		\draw[->](1,0)--(16.0,0);
		\draw[->](1,3)--(16.0,3);
		\point{2}{3}
		\point{10}{3}
		\point{11.5}{3}
		\point{14}{3}
		\point{15.5}{3}
		\draw(2, 3.5) node{$p_1$};
		\draw(10, 3.5) node{$p_2$};
		\draw(11.5, 3.5) node{$p_3$};
		\draw(14, 3.5) node{$p_4$};
		\draw(15.5, 3.5) node{$p_5$};
		\point{1.5}{0}
		\point{4.5}{0}
		\point{6}{0}
		\point{7.5}{0}
		\point{13}{0}
		\point{15.5}{0}
		\draw(1.5, -0.5) node{$q_1$};
		\draw(4.5, -0.5) node{$q_2$};
		\draw(6, -0.5) node{$q_3$};
		\draw(7.5, -0.5) node{$q_4$};
		\draw(13, -0.5) node{$q_5$};
		\draw(15.5, -0.5) node{$q_6$};
	\end{tikzpicture}
	\caption{\label{fig:perm_complexity} The original forbidden areas of $I$ can be recovered from $I_1'=\{(p_1,p_1,\mtt{h}),(p_1,p_2,\mtt{-}),(p_3,p_3,\mtt{t}),(p_4,p_4,\mtt{t})\}$ and $I_2'=\{(q_1,q_1,\mtt{-}),(q_1,q_2,\mtt{-}),(q_4,q_4,\mtt{h}),(q_5,q_5,\mtt{t})\}$, where $\mtt{h}$ and $\mtt{t}$ represent a hub and a terminal, respectively.}
\end{figure}
\begin{algorithm2e}[t!]
	\caption{Reconstruction of $I$\label{alg:reconstruction}}
	\SetVlineSkip{0.5mm}
	$I \leftarrow \emptyset$ and $\msf{ActiveHub} \leftarrow \mtt{Null}$\;
	\Repeat{the stacks $I_1'$ and $I_2'$ are empty}{%
		\If{$\msf{ActiveHub} = \mtt{Null}$}{%
			$(p_1,q_1,r_1) \leftarrow I_1'.\mtt{pop}$ and $(p_2,q_2,r_2) \leftarrow I_2'.\mtt{pop}$\;
			$I \leftarrow I \cup \{(p_1,q_1,p_2,q_2)\}$\;
			\lIf{$r_1 = \mtt{h}$}{$\msf{ActiveHub} \leftarrow (1,p_1,q_1)$}
			\lElseIf{$r_2 = \mtt{h}$}{$\msf{ActiveHub} \leftarrow (2,p_2,q_2)$}
		}\ElseIf{$\msf{ActiveHub} = (1,p_1,q_1)$ for some $p_1,q_1$}{
			$(p_2,q_2,r_2) \leftarrow I_2'.\mtt{pop}$\;
			$I \leftarrow I \cup \{(p_1,q_1,p_2,q_2)\}$\;
			\lIf{$r_2 = \mtt{h}$}{$\msf{ActiveHub} \leftarrow (2,p_2,q_2)$}
			\lElseIf{$r_2 = \mtt{t}$}{$\msf{ActiveHub} \leftarrow \mtt{Null}$}
		}\ElseIf{$\msf{ActiveHub} = (2,p_2,q_2)$ for some $p_2,q_2$}{
			$(p_1,q_1,r_1) \leftarrow I_1'.\mtt{pop}$\;
			$I \leftarrow I \cup \{(p_1,q_1,p_2,q_2)\}$\;
			\lIf{$r_1 = \mtt{h}$}{$\msf{ActiveHub} \leftarrow (1,p_1,q_1)$}
			\lElseIf{$r_1 = \mtt{t}$}{$\msf{ActiveHub} \leftarrow \mtt{Null}$}
	}	}
\end{algorithm2e}

A point $p$ may appear in $I_i$ in the following ways:
\begin{enumerate}
	\item ending a forbidden interval started earlier: $(q,p) \in I_i$ for some $q <_{\pi_i} p$,
	\item starting and ending a minimal forbidden interval: $(p,p) \in I_i$, which may be a hub or a terminal,
	\item starting a new forbidden interval: $(p,q) \in I_i$ for some $q >_{\pi_i} p$, which may be a terminal but not a hub.
\end{enumerate}
We do not care whether or not it is a terminal in the first case, leaving counting different possibilities to the interval starting point $q$.
%Then, each $p \in X_s \cup \{\bot\}$ has at most 13 ways to occur in respective $I_i$.
%All in all, we have at most $13^k \times 13^k \le 2^{7.41k}$ possibilities for $I$.
Solving the recurrence equations
\begin{align*}
	A_{\mrm{b}}(n+1) &= 4 A_\mrm{b}(n) + 4 A_\mrm{u}(n)\,,	&	A_{\mrm{b}}(0)=4\,, 
\\	A_{\mrm{u}}(n+1) &= 8 A_\mrm{b}(n) + 9 A_\mrm{u}(n)\,, &	A_{\mrm{u}}(0)=3\,,
\end{align*}
we obtain $A_{\mrm{b}}(n) \in O(2^{3.67 n})$.
All in all, we have at most $N \in O(2^{7.34k})$ possibilities for $I$.

Arguments similar to the proof of Theorem~\ref{thm:interval} derive the theorem.
\end{proof}

%\begin{proof}
%Let $k$ be the maximum size of the assigned set $X_s$ to a node of the tree-decomposition.
%There are at most $(k!)^2$ variants of $\pi$ of abstractions $(\pi,I,c)$.
%There are at most $2^{k^4}$ variants of $I \subseteq X_s^4$, but we can give a tighter upper bound $|I| \le 2^{7.41k}$ in Appendix.
%Recall that $\mscr{I}_s$ is reduced, where if $\mscr{I}_s$ has two elements of the form $ (\pi,I,c)$ and $(\pi,I,c')$, then $c=c'$.
%Therefore, each $\mscr{I}_s$ may contain at most $N = (k!)^2 \cdot 2^{7.41k}$ elements.
%To calculate $\mscr{I}_s$ from $\mscr{I}_t$ for children $t$ of $s$, it takes $O(N^2 \mrm{poly}(k))$ time for some polynomial function $\mrm{poly}$, since it has at most two children.
%All in all, our algorithm runs in $O(|V| N^2 \mrm{poly}(k))$ time.
%If the tree-decomposition has no join nodes, and thus is a path decomposition, then it will run in $O(|V|N \mrm{poly}(k))$ time.
%% where $N' = (k'!)^2 \cdot 2^{k^{\prime 2}}$.
%\end{proof}

%% file: interval_variants.tex
% !TEX root = main_arxiv.tex
\section{Other Classes Related to Interval Graphs}
The algorithm presented in the previous section can be applied to the $\mcal{C}$-\textsc{Edge-Deletion} for some subclasses and a superclass of interval graphs with a slight modifications.

\subsection{Proper interval graphs}
An interval representation $\pi$ is said to be \emph{proper} if there are no $u,v \in V$ such that $l_u <_\pi l_v <_\pi r_v <_\pi r_u$.
An interval graph is \emph{proper} if it admits a proper interval representation.
In accordance with the definition of the graph class, we simply require $\pi'$ in \textit{Introduce Node} of the algorithm in Section~\ref{sec:interval} to be a proper interval representation.
\begin{corollary}
The edge deletion problem for proper interval graphs can be solved in $O(|V| N^2 \mrm{poly}(k))$ time 
where $N=2^{3.91k}\frac{(2k)!}{(k+1)!}$ for the treewidth $k$ of $G$.
If $k$ is the pathwidth, it can be solved in $O(|V| N \mrm{poly}(k))$ time.
\end{corollary}
\begin{proof}
Let $k$ be the maximum size of the assigned set $X_s$ to a node of a nice tree-decomposition.
The number of possible proper interval representations over $k$ vertices is $k! C_k$ for the $k$-th Catalan number $C_k=\frac{(2k)!}{(k+1)!k!}$.
Let $A(n)$ be the numbers of possible ways of drawing forbidden intervals with $n$ points in addition to $\bot$ where the last intervals may be right-unbounded.
We observed in the proof of Theorem~\ref{thm:interval} (Figure~\ref{fig:interval_complexity}) that $A(n+1) \le 5A(n)$.
However, in the case of proper interval representations, we never have $(l_u,l_u) \in I$.
If the last $(n+1)$st point is $l_u$, we have at most three times more possible forbidden interval sets than the number of possible forbidden interval sets over $n$ points.
% $N \in O(2^{4.38k})$
Therefore, the number of possible abstractions in each $\mscr{I}_s$ is at most
% $O(\frac{(2k)!}{(k+1)!} \cdot 2^{4.38 k})$.
$O(\frac{(2k)!}{(k+1)!} \cdot 3^{k}\cdot 5^k) \subseteq O(2^{3.91k}\frac{(2k)!}{(k+1)!} )$.
\end{proof}
% log_2(15) < 3.9069

\subsection{Trivially perfect graphs}
An interval representation $\pi$ is said to be \emph{nested} if there are no $u,v \in V$, such that $l_u <_\pi l_v <_\pi r_u <_\pi r_v$.
A \emph{trivially perfect graph} (a.k.a.\ \emph{nested interval graph}) is an interval graph that admits a nested interval representation.
The algorithm presented in Section~\ref{sec:interval} can easily be modified so that it solves the edge deletion problem for trivially perfect graphs.
In accordance with the definition of the graph class, we simply require $\pi'$ in \textit{Introduce Node} of the algorithm in Section~\ref{sec:interval} to be a nested interval representation.
\begin{corollary}
The edge deletion problem for trivially perfect graphs can be solved in $O(|V| N^2 \mrm{poly}(k))$ time 
where $N=2^{2.4k}\frac{(2k)!}{k!}$ for the treewidth $k$ of $G$.
If $k$ is the pathwidth, it can be solved in $O(|V| N \mrm{poly}(k))$ time.
\end{corollary}
\begin{proof}
In this case, forbidden intervals and active intervals are nested.
Let $D_n$ be the number of possible pairs $(\pi,I)$ in abstractions with $n$ active vertices modulo renaming of vertices.
Then, the maximum cardinality of $\mscr{I}_s$ is at most $n! \cdot D_n$.
Hereafter we assume that neither $(\bot,\bot)$ nor $(\pred{\pi}(\top),\pred{\pi}(\top))$ appears in $I$.
Let us say that $(\pi,I)$ is \emph{splittable} if $\mathit{LR}_{X_s} - \{\bot,\top\}$ can be partitioned into two non-empty subsets $Y_1$ and $Y_2$ so that
\begin{itemize}
\item $\max_{\pi} Y_1 < \min_{\pi} Y_2$,
\item for any $u \in X_s$, $l_u$ and $r_u$ belong to the same component $Y_1$ or $Y_2$,
\item if $(p,q) \in I$ and $p \neq q$, $\suc{\pi}(p)$ and $q$ belong to the same component $Y_1$ or $Y_2$.
\end{itemize}
Otherwise, $(\pi,I)$ is \emph{non-splittable}.
Let $D'_n$ and $D''_n$ count the numbers of non-splittable and splittable pairs $(\pi,I)$, respectively, where $n=|X_s|$.
%($D'_n$: unsplittable. $D''_n$: splittable.)
We have $D_0=D_0'=1$,  $D_0''=D_1''=0$, $D_1=D_1' = 3$, and for $n \ge 2$,
\begin{align}
	D_n &= D'_n + D''_n\,,
\\ \label{eq:nonsplit}
	D'_n &= C_n + 4D_{n-1}\,,
\\ \label{eq:split}
	D''_n &= \sum_{i=1}^{n-1} 2 D'_i D_{n-i}\,.
\end{align}
The first term $C_n$ of (\ref{eq:nonsplit}) is the $n$th Catalan number, which counts the possible permutations $\pi$ when $(\bot,\pred{\pi}(\top)) \in I$, i.e., when $I = \{ (\bot,\pred{\pi}(\top)) \}$.
The second term $4D_{n-1}$ counts cases where $(\bot,\pred{\pi}(\top)) \notin I$ and $(\suc{\pi}(\bot),\pred{\pi}(\top))=(l_u,r_u)$ for some vertex $u$.
The coefficient $4$ counts the number of subsets of $\{(l_u,l_u),(\pred{\pi}(r_u),\pred{\pi}(r_u))\}$.
Note that when $n=1$, $l_u = \pred{\pi}(r_u)$ and thus we have $D_1' = C_1 + 2 D_0$.
Each term $2 D'_i D_{n-i}$ of (\ref{eq:split}) counts the number of pairs where the first splitting point is at the $i$th active vertex interval.
In other words, $2i$ is the minimum cardinality of the first splitting component $Y_1$. 
The coefficient $2$ counts the cases where $(r_v,r_v)$ presents and absents in $I$ where $r_v = \max_{\pi} Y_1$.
%\begin{align}
%	D_n &= D'_n + D''_n\,,
%\\ \label{eq:nonsplit}
%	D'_n &= 4 C_n + 4D_{n-1}\,,
%\\ \label{eq:split}
%	D''_n &= \sum_{i=1}^{n-1} D'_i D_{n-i}/2\,.
%\end{align}
%The first term $4C_n$ of (\ref{eq:nonsplit}) counts cases where we have $(\bot,\pred{\pi}(\top)) \in I$, in which case $I \subseteq \{(\bot,\bot),(\bot,\pred{\pi}(\top)),(\pred{\pi}(\top),\pred{\pi}(\top))\}$.
%The second term $4D_{n-1}$ counts cases where $(\bot,\pred{\pi}(\top)) \notin I$ and $(\succ{\pi}(\bot),\pred{\pi}(\top))=(l_u,r_u)$ for some vertex $u$.
%In both terms, the coefficient $4$ counts the number of subsets of $\{(\bot,\bot),(\pred{\pi}(\top),\pred{\pi}(\top))\}$.
%Each term $ D'_i D_{n-i}/2$ of (\ref{eq:split}) counts the number of pairs where the first splitting point is at the $i$th vertex.

We show $D_n \le 2^{\alpha n}C_n $ for $\alpha=2.4$ by induction on $n$.
One can confirm it is true for $n \le 22$ by calculation.
For $n \ge 23$,
\begin{align*}
	D_n &= C_n + 4D_{n-1} + 2\sum_{i=1}^{n-1} (C_i + 4D_{i-1}) D_{n-i} - 4D_{0} D_{n-1}
\\	&\le   C_n +  2 \sum_{i=1}^{n-1} (C_i + 2^{\alpha(i-1)+2}C_{i-1}) 2^{\alpha (n-i)}C_{n-i}
\\	& =  C_n + 2 \sum_{i=1}^{n-1} 2^{\alpha(n-i)} C_i C_{n-i} + 2^{\alpha(n-1)+3} \sum_{i=1}^{n-1} C_{i-1} C_{n-i}
\\	& =  C_n  + \sum_{i=1}^{n-1} (2^{\alpha(n-i)}+2^{\alpha i}) C_i C_{n-i} + 2^{\alpha(n-1)+3}(C_0C_{n-1} + C_{n-1})
\\	& \le  C_n 
	+ (2^{\alpha(n-1)}+2^{\alpha}) C_n
	+ 2^{\alpha(n-1)+4} C_{n-1} 
\\	& = (1+2^{\alpha(n-1)} + 2^{\alpha} + {\textstyle \frac{n+1}{4n-2}}2^{\alpha(n-1)+4})C_{n}\,.
\end{align*}
Noting that $n \ge 23$ and $\alpha = 2.4$, we obtain
\begin{align*}
	D_n &\le 2^{2.4 n} C_n\,.
\end{align*}
Hence, the number of possible abstractions is at most $(k+1)! D_{k+1} \le (k+1)! \cdot 2^{2.4 (k+1)} \frac{(2(k+1))!}{(k+2)!(k+1)!} \in O(2^{2.4 k} \frac{(2k)!}{k!})$.
\end{proof}

% Stirling's approximation: n! = \Theta(\sqrt{n} (n/e)^n)

\subsection{Circular-arc graphs}
Circular-arc graphs are a generalization of interval graphs which have an arc model, which can be seen as a ``circular'' interval representation.
For a linear order $\pi$ over a set $S$ and four elements $p_1,q_1,p_2,q_2 \in S$, we write $\its{\pi}{(p_1,q_1)}{(p_2,q_2)}$ if either
\begin{itemize}
	\item $p_1 <_\pi m <_\pi q_1$ for some $m \in \{p_2,q_2\}$, or
	\item $q_1 <_\pi p_1 <_\pi m$ or $m <_\pi q_1 <_\pi p_1$ for some $m \in \{p_2,q_2\}$.
\end{itemize}
Otherwise, we write $\prl{\pi}{(p_1,q_1)}{(p_2,q_2)}$.
A \emph{circular-arc graph} is a graph $G_\pi = (V,E)$ such that
\[
	E = \{\, \{u,v\} \subseteq V \mid \its{\pi}{(l_u,r_u)}{(l_v,r_v)} \,\}
\]
for some linear order $\pi$ over $\mathit{L}_V \cup \mathit{R}_V$. Note that this set contains neither $\top$ nor $\bot$.
The algorithm presented in Section~\ref{sec:interval} can easily be modified so that it solves the edge deletion problem for circular-arc graphs by replacing the definitions of $\prl{\pi}{}{}$ and $\its{\pi}{}{}$ as above,
and defining $\suc{\pi}( \max_\pi X)=\min_\pi X$ and $\pred{\pi}(\min_\pi X )=\max_\pi X$.
Since we allow $r_u <_\pi l_u$, the number of admissible arc model is bigger than that of (ordinary) interval representations.
This affects the computational complexity.
\begin{corollary}
The edge deletion problem for circular-arc graphs can be solved in $O(|V| N^2 \mrm{poly}(k))$ time 
where $N=(2k)! \cdot 2^{4.38 k}$ for the treewidth $k$ of $G$.
%where $N=(2k)! \cdot 2^{4.14 k} \in (2k)! \cdot 2^{O(k)}$ for the treewidth $k$ of $G$. [WHY 4.14?]
If $k$ is the pathwidth, it can be solved in $O(|V| N \mrm{poly}(k))$ time.
\end{corollary}
\begin{proof}
There can be $(2k+2)!$ varieties of $\pi$.
One can argue that each $\pi$ has at most $O(2^{4.38 k})$ sets of forbidden intervals similarly to the proof of Theorem~\ref{thm:interval}.
\end{proof}

\iffalse
A circular-arc graph $G_\pi$ is \emph{proper} if $\pi$ has no nested arcs.
That is, there are no $u,v \in V$ such that $l_u <_{\pi_{l_u}} l_v <_{\pi_{l_u}} r_v <_{\pi_{l_u}} r_u$ where $\pi_{l_u}$ is obtained from $\pi$ by ``shifting'' the minimum point to $l_u$: $p <_{\pi_u} q$ iff either $l_u \le_\pi p <_\pi q$, $ p <_\pi q <_\pi l_u$, or $q <_\pi l_u \le_\pi p$.
\begin{corollary}
The edge deletion problem for proper circular-arc graphs can be solved in $O(|V| N^2 \mrm{poly}(k))$ time 
where $N=XXX$ for the treewidth $k$ of $G$.
If $k$ is the pathwidth, it can be solved in $O(|V| N \mrm{poly}(k))$ time.
\end{corollary}
\fi

%% file: threshold.tex
% !TEX root = main_arxiv.tex
%\noindent {\bf Finding a Largest Threshold Subgraph: }
\subsection{Threshold graphs}
Threshold graphs are special cases of trivially perfect graphs, which can be defined in several different ways.
Here we use a pair of a vertex subset $W \subseteq V$ and a linear order $\pi$ over $\mit{R}_V$ as a \emph{threshold interval representation}.
We say that vertices $u$ and $v$ \emph{intersect} on $(W,\pi)$ if and only if $u \in W$ and $r_{v} <_\pi r_{u}$ or the other way around.
% $l_u \le_\pi r_u$ for all $u \in X_1$ and $l_v \le_\pi p$ for all $v \in X_2$ and $p \in \mit{LR}_X$ for some disjoint subsets $X_1,X_2$ of $X$ with $X= X_1 \cup X_2$.
A \emph{threshold graph} is a graph $G_{W,\pi} = (V,E_{W,\pi})$ where $(W,\pi)$ is a threshold interval representation on $V$ and
%\[
$E_{W,\pi} = \{\, \{u,v\} \subseteq V \mid \text{$u$ and $v$ intersect on $(W,\pi)$} \,\}$.
%\,.\]
By extending $\pi$ to $\pi'$ over $\mathit{LR}_V$ so that $l_w <_{\pi'} r_v$ for all $w \in W$ and $v \in V$ and $\suc{\pi'}(l_u) = r_u$ for all $u \in V - W$, then the induced interval graph coincides with the threshold graph.
To attain drastic improvement on the complexity, we design an algorithm for the edge deletion problem for threshold graphs from scratch, rather than modifying the one for interval graphs.

%\subsection{Algorithm invariant}
For a threshold representation $(Y,\rho)$ of a subgraph $G_{Y,\rho}=(X_{\le s},E_{Y,\rho})$, we define its abstraction $\mscr{A}((Y,\rho),s) = (Y',\pi,b,p,c)$ as follows:
%% \begin{itemize}
%% 	\item $\pi$ is the restriction of $\rho$ to $R_{X_s}$,
%% 	\item $Y' = Y \cap X_s$,
%% 	\item if $Y'=Y$, then $b=0$ and 
%% \(
%% p = \max_{\pi}\{\, p \in R_{X_s} \mid p <_\rho r_y \text{ for all } y \in X_{\le s} - X_s \,\}
%% ,\)
%% 	\item if $Y' \neq Y$, then $b=1$ and 
%% \(
%% p = \max_{\pi}\{\, p \in R_{X_s} \mid p <_\rho r_y \text{ for some } y \in Y \setminus Y' \,\}
%% \,.\)
%% 	\item $c = |E_{\le s} - E_{Y,\rho} - E_{s}|= |\{\, \{u,v\} \in E \mid \{u,v\} \nsubseteq X_s \text{ and $u$ and $v$ do not intersect on $(Y,\rho)$} \,\}|$.
%% \end{itemize}
(1) $\pi$ is the restriction of $\rho$ to $R_{X_s}$,
(2) $Y' = Y \cap X_s$,
(3) if $Y'=Y$, then $b=0$ and
\(
p = \max_{\pi}\{\, p \in R_{X_s} \mid p <_\rho r_y \text{ for all } y \in X_{\le s} - X_s \,\}
,\)
(4) if $Y' \neq Y$, then $b=1$ and 
\(
p = \max_{\pi}\{\, p \in R_{X_s} \mid p <_\rho r_y \text{ for some } y \in Y - Y' \,\}
\,\), and 
(5) $c = |E_{\le s} - E_{Y,\rho} - E_{s}|= |\{\, \{u,v\} \in E \mid \{u,v\} \nsubseteq X_s \text{ and $u$ and $v$ do not intersect on $(Y,\rho)$} \,\}|$.
  
We say that $(Y',\pi',b',p',c')$ \emph{dominates} $(Y,\pi,b,p,c)$ if
 $Y' = Y$, $\pi'=\pi$, $c' \le c$, and either 
(a) $b'=b=0$ and $p' \ge_{\pi} p$,
(b) $b'=b=1$ and $p' \le_{\pi} p$, or
(c) $b'=0$ and $b=1$.
%(Note: when $b=1$, the order amongst $r_u$ bigger than $\max\{\,r_y \in R_{Y} \mid y \in Y\,\}$ does not matter.)
%To holds the threshold counterpart of Condition~\ref{cond:invariant_i}
Using the above invariant, we can provide an algorithm for \textsc{Threshold-Edge-Deletion}. % and describe the details of the algorithm in Appendix A.2. 

Our algorithm assigns a set $\mscr{I}_s$ for each node $s$ of $T$ so that the threshold counterpart of Condition~\ref{cond:invariant_i} holds.
Accordingly $\mscr{I}_s =  \{ (\emptyset, o, 0, \bot, 0) \}$ for leaf nodes $s$.

\paragraph*{Introduce Node:}
Suppose $s$ has just one child $t$ such that $X_{s} = X_t \cup \{x\}$.
For each $(Y,\pi,b,p,c) \in \mscr{I}_t$, we add to $\mscr{I}'_s$ all tuples $(Y',\pi',b,p',c)$ fulfilling the following conditions:
\begin{itemize}
	\item $\pi'$ is an extension of $\pi$ to $R_{X_s}$,
	\item $Y \subseteq Y' \subseteq Y \cup \{x\}$.
	\item if $\{x,u\} \notin E$ for $u \in X_{t}$, then $x$ and $u$ do not intersect in $(Y',\pi')$.
%	\begin{itemize}
%		\item $x,u \notin Y'$,
%		\item $x \in Y'$, $u \notin Y'$ and $r_x <_{\pi'} r_u$, or
%		\item $x \notin Y'$, $u \in Y'$ and $r_u <_{\pi'} r_x$,
%	\end{itemize}
	\item if $b = 0$, then $x \notin Y'$ or $r_x <_{\pi'} \suc{\pi}(p)$,
	\\	and moreover $p' = \begin{cases} r_x & \text{if $p <_{\pi'} r_x <\suc{\pi}(p)$,}
		\\	p & \text{otherwise,}	 \end{cases}$
	\item if $b = 1$, then $x \notin Y'$ and $p' = p <_{\pi'} r_x$.
\end{itemize}
We then obtain $\mscr{I}_{s}$ by reducing $\mscr{I}'_{s}$.

\paragraph*{Forget Node:}
Suppose $s$ has just one child $t$ such that $X_{t} = X_s \cup \{x\}$.
For each $(Y,\pi,b,p,c) \in \mscr{I}_t$, we add the tuple $(Y',\pi',b',p',c')$ to $\mscr{I}'_s$ where
\begin{itemize}
	\item $\pi'$ is the restriction of $\pi$ for $R_{X_s}$,
	\item $Y' = Y - \{x\}$,
	\item $b'=1$ if $x \in Y$, and $b'=b$ otherwise,
	\item if $b'=0$, then $p'= \min\{p, \pred{\pi}(r_x)\}$,
	\item if $b'=1$, then
	$p' = \begin{cases} \pred{\pi}(r_x) & \text{if $b=0$ or $p <_\pi r_x \wedge x \in Y$ or $r_x = p$},
	\\	p	&\text{otherwise.}\end{cases}$
%	\item if $b'=1$ and $b=0$, then	$p' = \pred{\pi}(r_x)$.
	\item $c'=c+|\{\, \{u,x\} \in E \mid \text{$u \in X_s$ and $x$ do not intersect on $(Y,\pi)$} \,\}|$
\end{itemize}
We then obtain $\mscr{I}_{s}$ by reducing $\mscr{I}'_{s}$.

\paragraph*{Join node:}
Suppose $s$ has two children $t_1$ and $t_2$, where $X_{s} = X_{t_1} = X_{t_2}$.
We add $(Y,\pi,b,p,c)$ to $\mscr{I}'_{s}$ if there are
 $(Y,\pi,b_1,p_1,c_1) \in \mscr{I}_{t_1}$ and $(Y,\pi,b_2,p_2,c_2) \in \mscr{I}_{t_2}$
such that $c=c_1+c_2$, and either
\begin{itemize}
\item $b = b_1 = b_2$ and $p=\min_{\pi_1}\{p_1,p_2\}$,
\item $b=b_1=1$, $b_2=0$ and $p = p_1 \le_{\pi} p_2$, or
\item $b=b_2=1$, $b_1=0$ and $p = p_2 \le_{\pi} p_1$.
%\item $b_1=0$ and $b=b_2=1$, then $p_2 \le_{\pi_1} p_1$.
\end{itemize}
%If they are compatible, we add to $\mscr{I}'_{s}$ the tuple $(\pi_1,Y_1,b,p,c_1+c_2)$ where
%\[
%(b,p)=\begin{cases}
%	(1,p_i) & \text{if $b_i=1$ for $i \in \{1,2\}$,}
%\\	(0,\min_{\pi_1}\{p_1,p_2\}) & \text{if $b_1=b_2=0$.}
%\end{cases}
%\]
We then obtain $\mscr{I}_{s}$ by reducing $\mscr{I}'_{s}$.

\begin{theorem}
The edge deletion problem for threshold graphs can be solved in $O(|V| N^2 \mrm{poly}(k))$ time 
where $N=k! \cdot 2^{k}$ for the treewidth $k$ of $G$.
If $k$ is the pathwidth, it can be solved in $O(|V| N \mrm{poly}(k))$ time.
\end{theorem}

%
%\begin{theorem}
%Threshold.
% Let $k$ be the maximum bag size.
%\begin{itemize}
%\item $|\mscr{I}_s| \le k! \cdot 2^{k} \cdot k \cdot 2 = (k+1)! \cdot 2^{k+1} = M$.
%\item Update: $O(M^2 \mathrm{poly}(k))$, (it will be $O(M' \mathrm{poly}(k'))$ for the path-width $k'$)
%\item Total: $O(nM^2 \mathrm{poly}(k))$.
%\end{itemize}
%\end{theorem}

%% file: conclusion.tex
% !TEX root = main.tex
\section{Conclusion}\label{sec:conc}
We have proposed FPT algorithms for \textsc{Edge-Deletion} to some intersection graphs parameterized by treewidth in this paper.
Our algorithms maintain partial intersection models on a node of a tree decomposition with some restrictions and extend the models consistently for the restrictions in the next step. 
We expect that the ideas in our algorithms can be applied to other intersection graphs whose intersection models can be represented as linear-orders, for example circle graphs, chain graphs and so on, and to \textsc{Vertex-Deletion} of intersection graphs. 

%open questions
We have the following questions as future work: 
\begin{itemize}
\item Do there exist single exponential time algorithms for the considered problems, that is, $O^*(2^{\tw(G)})$ time, or can we show matching lower bounds assuming the Exponential Time Hypothesis?
\item Are there FPT algorithms parameterized by treewidth for $\C$-\textsc{Completion} which is to find the minimum number of adding edges to obtain a graph in an intersection graph class $\C$?
  We can naturally apply the idea of our algorithms to $\C$-\textsc{Completion} problems. 
  While $\C$-\textsc{Edge-Deletion} algorithms do not allow introduced objects to intersect with forgotten objects, $\C$-\textsc{Completion} algorithms do allow it with the cost of addition of new edges.
  Thus $\C$-\textsc{Completion} algorithms based on this naive approach will be XP algorithms since we have to remember the number of forgotten objects in the representation to count the number of intersections between the introduced objects and forgotten objects.
  %It is interesting to study the problem to add a minimum number of edges to obtain a graph 
%  Such completion problems are equivalent to the edge deletion problems for the complements of the input graph and the complement of the target graph class.
%%%%% Permutation
%%   Another approach may be based on the fact that $\C$-\textsc{Completion} is equivalent to $\mathcal{D}$-\textsc{Edge-Deletion} for $\mathcal{D}$ the class of the complements of graphs in $\C$, where the input is also complemented.
%%   It is known that complements of permutation graphs are also permutation graphs so our \textsc{Permutation-Edge-Deletion} algorithm works for \textsc{Permutation-Completion} but it is not an FPT algorithm in general since the treewidth of the complement of a graph with small treewidth can be very large. 
%% \ryoshi{PERMUTATION}
  %We would like to consider FPT algorithms for Graph sandwich problems for geometric intersection graphs. 
\item Are there FPT algorithms for \textsc{Edge-Deletion} to intersection graphs defined using objects on a plane, like unit disk graphs?
The intersection graph classes discussed in this paper are all defined using objects aligned on a line.
Going up to a geometric space of higher dimension is a challenging topic. 
%  It is difficult to represent a disk model as the constant number of permutations~\cite{McDiarmidM13} and is a well-known question whether or not there is $O(n\lg{n})$ space representation for (unit) disk intersection graphs.
\end{itemize}